\documentclass[11pt]{article}
\usepackage[margin=1in]{geometry}

\newcommand{\ignore}[1]{{}}

\usepackage{amsmath,amssymb,amsthm}
\usepackage{dsfont}
\usepackage{natbib}
\usepackage{graphicx} 
\usepackage[backref=page,linktocpage=true,breaklinks,colorlinks,citecolor=blue,linkcolor=BrickRed]{hyperref}
\usepackage{cleveref}
\crefname{algocf}{algorithm}{algorithms}
\Crefname{algocf}{Algorithm}{Algorithms}
\usepackage{enumitem} 
\usepackage{float}
\usepackage{bm}
\usepackage{tikz}
\usetikzlibrary{shapes,positioning}
\usepackage{thmtools,thm-restate}
\usepackage[ruled,linesnumbered]{algorithm2e}
\usepackage[dvipsnames]{xcolor}

\usepackage{decorated}
\usepackage{tocloft}
\usepackage{notation}
\usepackage{mathtools}
\usepackage{tcolorbox}
\usepackage{nicefrac}
\newtcolorbox{lpbox}{
  colframe=black,
  boxrule=0.8pt,
  arc=2pt,
  left=6pt,
  right=6pt,
  top=6pt,
  bottom=6pt
}

\newcommand{\snote}[1]{{\color{blue}{SAHIL: {#1}} }}

\newcommand{\rev}{{\rm Rev}}

\newcommand{\rem}{\lambda}
\newcommand{\good}{Scaling\;}
\newcommand{\f}{\mathsf{s}}

\newcommand{\vc}[1]{\mathbf{#1}}
\newcommand{\bx}{\vc{x}}
\newcommand{\by}{\vc{y}}
\newcommand{\bz}{\vc{z}}

\def\E{{\mathbb{E}}}		

\def\supp{\operatorname{supp}}
\def\1{\mathds{1}}             
\allowdisplaybreaks[4] 

\newtheorem{theorem}{Theorem}[section]
\newtheorem{claim}[theorem]{Claim}
\newtheorem{proposition}[theorem]{Proposition}
\newtheorem{lemma}[theorem]{Lemma}

\theoremstyle{definition}

\newtheorem{definition}[theorem]{Definition}

\newcommand{\hide}[1]{}

\definecolor{toc}{RGB}{13,55,174}	
\usepackage{hyperref}				
\hypersetup{
colorlinks=true,
citecolor=toc,
filecolor=black,
linkcolor=toc,
urlcolor=toc
}

\usepackage{cleveref}
\crefname{algocf}{algorithm}{algorithms}
\Crefname{algocf}{Algorithm}{Algorithms}

\usepackage{xspace}

\newcommand{\OWDI}{\ensuremath{\mathsf{OWDI}}\xspace}
\newcommand{\IGNORE}[1]{}
\renewcommand{\emph}[1]{\textit{#1}}
\title{Philosopher and Prophet Inequalities for Divisible Items}

\author{}

\begin{document}

	\author{Thiago Oliveira\thanks{
         (toliveira9@gatech.edu)
         H. Milton Stewart School of Industrial and Systems Engineering,
        Georgia Tech. Supported in part by NSF awards CCF-2440113, 2504994, 2106444 and the 2026 Spring ARC-ACO Fellowship.
    }
	\and Mohit Singh\thanks{
         (mohit.singh@isye.gatech.edu)
         H. Milton Stewart School of Industrial and Systems Engineering,
        Georgia Tech.
        Supported in part by NSF awards CCF-2504994, 2106444.
    }
	\and Sahil Singla\thanks{
        (ssingla@gatech.edu)
        School of Computer Science,
        Georgia Tech.
        Supported in part by NSF awards CCF-2327010 and CCF-2440113.
        }
}

\maketitle


\begin{abstract}


We study online welfare maximization with divisible resources. A sequence of $n$ players arrive one by one; upon arrival, each player draws a valuation function over $m$ divisible items from a known distribution, reveals this valuation, and must be allocated an irrevocable fractional bundle subject to unit supply constraints. While online welfare maximization has been extensively studied for indivisible items and combinatorial valuations, much less is known when the resources are divisible and players have multi-dimensional concave valuations.

\smallskip
We give approximation algorithms for monotone concave valuations satisfying
diminishing returns. Our main result is a \(2/3\)-approximation to the optimal
online policy, also known as the philosopher benchmark. The algorithm is guided
by a low-dimensional concave relaxation of the online benchmark and rounds it
via a new single-item capped online contention resolution scheme. This
\car problem allocates to each realized type no more than its prescribed
fractional bundle while preserving a \(2/3\)-fraction of that bundle in
expectation. Its analysis uses a submartingale potential for the remaining
supply. On the hardness side, we show that computing the optimal online policy is \#P-hard even for a single divisible item.

\smallskip
We also obtain a tight prophet inequality against the offline hindsight
optimum. We show that a fixed-price auction with one linear per-unit price for
each original divisible item achieves a \(1/2\)-approximation to the
offline/prophet benchmark. The prices are obtained by aggregating
Aumann--Shapley supporting prices, a continuous analogue of supporting prices
for submodular/XOS set functions, and yield simple item prices rather than
copy-dependent prices arising from discretization. The factor \(1/2\) for the prophet benchmark is information-theoretically tight even for one item with linear valuations.

\end{abstract}

\newpage
\section{Introduction}

Online welfare maximization asks how to allocate limited resources to agents
arriving over time. In its Bayesian form, player \(i\in[n]\) draws a valuation
\(V_i\sim\Dcal_i\), reveals it upon arrival, and must be allocated an
irrevocable bundle from the remaining resources. The goal is to maximize
expected social welfare.
This model is a common abstraction in online algorithms and also underlies
static posted-price mechanisms in online market design
\cite{EIV-Book23,Lucier17,FeldmanGL15a,DuttingFKL20a}.

Most of the literature on online welfare maximization focuses on
\emph{indivisible} items: each item can be allocated to at most one player, and
players value subsets of items, e.g., via additive, submodular, or XOS
valuations. 
For example, fixed posted prices yield a $1/2$-approximation to the offline (a.k.a.\ prophet) benchmark for submodular (or even XOS) valuations \cite{FeldmanGL15a}, and recent work gives improved guarantees against the online (a.k.a.\ philosopher) benchmark for unit-demand and submodular valuations \cite{PapadimitriouPSW21a,BravermanDPSW25a,STW-SODA26}.

However, many modern resources are naturally \emph{divisible}---e.g., CPU time, bandwidth, energy, or budget---and should be allocated fractionally. This motivates us to study the problem of maximizing \emph{Online Welfare with Divisible Items} (\OWDI): there are $m$ divisible items, each with unit supply\footnote{The unit-supply assumption is without loss of generality after rescaling.}, and players arrive sequentially. When player $i$ arrives, the algorithm observes \(V_i\sim\Dcal_i\) and must irrevocably choose a fractional bundle $\bx_i\in[0,1]^m$ subject to the per-item capacity constraints $\sum_i x_{i,j}\le 1$ for all $j\in[m]$. We formally define the model in \Cref{sec:prelim}.


We focus on monotone concave DR-submodular valuations, a standard continuous
counterpart of submodular set functions. A valuation
$f:[0,1]^m\to\Reals_+$ is DR-submodular if it satisfies diminishing returns:
for all $\vc{x},\vc{y},\vc{z}\in\Reals_+^m$ with $\vc{x}\le \vc{y}$ and with
$\vc{x}+\vc{z},\vc{y}+\vc{z}\in[0,1]^m$,
\[
    f(\vc{x}+\vc{z})-f(\vc{x})
    \;\ge\;
    f(\vc{y}+\vc{z})-f(\vc{y}).
\]
Together with concavity and monotonicity, this class captures natural
continuous substitutes for submodular valuations, including budget-additive and
separable concave valuations. 
Thus, \OWDI is a natural continuous analogue of online welfare maximization with submodular agents. 
DR-submodular valuations have been extensively studied in online resource allocation with concave objectives, 
typically as a generalization of online bipartite matching \cite{KarpVV90a} and AdWords \cite{MehtaSVV07a} in models where \emph{items} arrive online, rather than buyers with multi-dimensional valuations \cite{Patton-SODA26,HJPSZ-FOCS24,DJ-STOC12,FKMMP-WINE09}.

We evaluate algorithms against two standard benchmarks. The \emph{offline}
benchmark, also called the \emph{prophet} benchmark, is the expected welfare of
the optimal allocation that observes the entire realized valuation profile in
hindsight. The \emph{online} benchmark, also called the \emph{philosopher}
benchmark, is the expected welfare of the optimal online policy with the same
information as the algorithm but with unlimited computational power. The
philosopher benchmark is always at most the prophet benchmark, but is often the
more tailored yardstick: it measures the loss from using an efficient
online algorithm rather than the unavoidable loss from not seeing the future.

Since the online submodular welfare results for the indivisible setting extend to the divisible setting by discretization, this raises the following basic question:
\vspace{-0.1cm}
\begin{quote}
    Is \OWDI simply a continuous version of the indivisible-item problem, 
or  does fractional allocation reveal new structure and lead to stronger,
    simpler algorithms? 
\end{quote}
In the indivisible setting, online algorithms must resolve contention among
discrete subsets of items. In the divisible setting, by contrast, one may hope to resolve contention
by assigning fractional items as resources are depleted while allocating them to arriving players. Are there new algorithmic techniques that can utilize this power to even go beyond what is attainable for the discrete setting?


\subsection{Our Contribution and Techniques}

Our main result is a philosopher inequality for divisible resources that achieves a $\frac23$-approximation. This exceeds the tight $(1-1/e)$ value-oracle barrier for offline submodular welfare with indivisible items \cite{MirrokniSV08a}, highlighting that fractional allocation enables genuinely stronger guarantees.

\begin{restatable}{theorem}{mainphi}
\label{teo:mainphi}
        For the online welfare with divisible items (\OWDI) problem with
        monotone concave DR-submodular valuations, there exists a
        polynomial-time algorithm achieving a $\frac23$-approximation to the
        online/philosopher benchmark.
\end{restatable}


The algorithm relies on three ingredients, each using the divisibility of the
resources and the concave DR-submodularity of the valuations.

\medskip \noindent \textbf{(1) Configuration LP collapses.} 
As in the indivisible philosopher-inequality literature, the natural starting
point is a configuration relaxation with online feasibility constraints \cite{STW-SODA26,PapadimitriouPSW21a}.  The configuration relaxation for the optimal online policy
has variables describing the distribution over \emph{bundles} of resources to players.
For divisible resources with concave valuations, this relaxation simplifies substantially. By concavity, a distribution over fractional bundles can be replaced by its mean bundle without decreasing the relaxed objective, and the online-feasibility constraints remain linear in these mean allocations. Thus the configuration relaxation collapses to a low-dimensional concave program whose variables are expected fractional allocations of resources to players. This compact relaxation is the object we round online. We describe the configuration relaxation, the concave relaxation, and their equivalence in \Cref{sec:prelim}.

\medskip \noindent \textbf{(2) Reducing to single-item \car.} 
We show that DR-submodularity allows us to reduce the multi-item rounding problem to a \emph{single} item problem with an additional capping constraint. In the single item problem, each realized player has a prescribed target amount. The goal is to give a fractional assignment of the items to the randomly arriving players that must satisfy two requirements simultaneously:
it should assign at least an $\alpha$-fraction of target in expectation, and it
should never allocate more than the target itself. We call such a procedure a \emph{capped online contention resolution scheme} (\car).

For the reduction, we run a separate instance of the one-item \car problem for each
item. The target per item for any player is exactly the fractional amount assigned to it by the concave relaxation. Using DR-submodularity, we combine the item-wise guarantees into the same $\alpha$-approximation for the player’s full value.

\medskip \noindent \textbf{(3) A scaling-based \car algorithm.}
The remaining task is to design a \car for a single
divisible item. 
Our \car algorithm is based on a continuous scaling approach. Let \(R_t\) denote the actual remaining
supply when player \(t\) arrives, and let \(\rem_t\) denote the remaining
supply predicted by the concave relaxation at the same point. The ratio
\(R_t/\rem_t\) measures whether the online process is ahead of or behind the
relaxation. When this ratio is small, the algorithm has spent more supply than expected and assigns only a smaller fraction than expected; when it is large, the algorithm can be more generous, but can never assign more than the target due to the capping constraint. Formally, the algorithm uses a \emph{scaling function}
$
    \f:\Reals_+\to[0,1],$
and allocates a scaled version of the current target with scaling factor
\(\f(R_t/\rem_t)\), subject to the natural caps imposed by the target and the
remaining supply. The condition \(\f\le 1\) ensures that no player receives more
than their prescribed target, while $\f(x)\leq x$ ensures that the algorithm never
allocates more than the remaining supply. We formally define a scaling function in \Cref{def:good}.

The analysis tracks the evolution of the ratio between the actual remaining
supply and the relaxation's predicted remaining supply. A potential-function based martingale argument shows that, as long as \(\f\)
satisfies a certain inequality, the worst preservation guarantee occurs at the
first player. Since initially \(R_1=\rem_1=1\), the approximation factor of the capped
scheme is exactly the initial scaling value \(\f(1)\). Thus the design problem reduces to choosing a scaling function $\f$ that respects the cap, satisfies the martingale inequality, and maximizes $\f(1)$. We exhibit a simple scaling function \(\f\) with \(\f(1)=2/3\), which yields the approximation guarantee in
\Cref{teo:mainphi}.  We remark that a heuristic computational search suggests that the optimal approximation via this approach is around $0.69$. We also show the integrality gap of the concave program is at most $0.75$ (see \Cref{sec:appendixIG}) but computational evidence suggests an improved bound of $0.71$.  


We also show in \Cref{sec:hardness} that the philosopher benchmark is
computationally nontrivial even in the divisible model: computing the optimal
online policy exactly is \(\#P\)-hard, even for a single divisible item. 
Thus the relaxation is not merely a notational convenience: it is an efficiently
computable surrogate for a benchmark that is itself hard to compute.

\medskip

Our second result is a tight prophet inequality via simple posted prices.

\begin{restatable}{theorem}{postedpricecomp}
\label{teo:postedpricecomp}
    For the online welfare with divisible items (\OWDI) problem with monotone concave DR-submodular valuations, there exist efficiently computable linear per-unit item prices such that the fixed-price auction achieves a \(1/2\)-approximation to the offline/prophet benchmark.
\end{restatable}

A key technical ingredient is a continuous analogue of \emph{supporting prices} for XOS/submodular set functions.
For a valuation and allocation vector \(\bx\), the Aumann--Shapley value
assigns coordinate prices by integrating the gradient along the ray from
\(0\) to \(\bx\) \cite{AumannShapley-Book75}. We aggregate these supporting prices across players and types using the offline relaxation. 
We show that for monotone DR-submodular valuations these prices support every coordinatewise dominated allocation $0\le \by\le \bx$ in the same sense as supporting prices for sets.
This lets us transfer the standard balanced-prices analysis of prophet inequalities to the divisible setting \cite{DuttingFKL20a,FeldmanGL15a}.
The main additional challenge is aggregation: Aumann--Shapley supporting prices are naturally defined separately for each player and realized type, while a posted-price mechanism must announce a single per-unit price for each item. We show that the appropriate welfare-weighted aggregation preserves the revenue--utility decomposition required by the balanced-prices proof. 
The resulting prices are computed from an optimal solution of the low-dimensional offline relaxation.

Divisibility is also crucial for this result.
Discretization can recover the factor $1/2$ using \cite{FeldmanGL15a}, but naturally produces copy-dependent supporting prices because a divisible item can be split among multiple buyers. Our analysis instead yields a single linear per-unit price for every original divisible resource.

\ignore{
We evaluate online performance against two standard benchmarks.
The \emph{offline} (or \emph{prophet}) benchmark is the expected welfare of the optimal allocation that sees the entire realized valuation profile in hindsight.
The \emph{online} (or \emph{philosopher}) benchmark is the expected welfare of the optimal \emph{online policy} with the same information as our algorithm but unlimited computation.
The online benchmark is smaller than the prophet benchmark, but more tailored: it measures the loss from computational/algorithmic restrictions rather than from lack of future information.


Given the rich literature on indivisible items, we ask whether the divisible
setting admits analogues of the two central guarantees, and whether divisibility
can be exploited to obtain stronger guarantees against the online benchmark:
\begin{quote}
Can we obtain a constant strictly better than $1/2$ against the
\emph{online/philosopher} benchmark for general concave diminishing-returns
valuations?
Against the stronger \emph{offline/prophet} benchmark, can the classical
$1/2$ posted-price guarantee for submodular/XOS valuations be extended to
divisible resources using \emph{simple linear prices}, namely a single
per-unit price for each item?
\end{quote}

Our main message is that divisibility is not merely a continuous limit of the indivisible setting. For the philosopher benchmark, divisibility collapses the natural configuration relaxation, including its online-feasibility constraints, to a low-dimensional concave program. This structure enables a new valuation-independent capped online rounding
scheme with a $(1-1/e)$ guarantee. For the prophet benchmark, divisibility raises a different challenge: a black-box discretization may produce copy-dependent prices, whereas we obtain one uniform linear per-unit price for each original divisible resource.

We focus on monotone concave DR-submodular valuations, the continuous analogue of monotone submodular set functions:
\[
f(\bx+\bz)-f(\bx)\ \ge\ f(\by+\bz)-f(\by)
\quad
\text{for \(\vc{x}, \vc{y}, \vc{z} \in \Reals^m_+\) with \(\vc{x} \leq \vc{y}\).}
\]
This class includes budget-additive and separable concave valuations, and captures diminishing marginal returns across resources.
}

\ignore{
\subsection{Our Contribution and Techniques}

Our main result is a philosopher inequality for divisible resources.
In \Cref{sec:phi}, we define the notion of a \good function; see \Cref{def:good}.
For
monotone concave DR-submodular valuations, we are able to give a constant factor
approximation to the optimal online policy depending on the value of a \good function.

\begin{restatable}{theorem}{mainphi}
\label{teo:mainphi}
        Let \(\f\) be a \good function. 
        For the online welfare with divisible items (\OWDI) problem with monotone concave DR-submodular valuations, there exists a polynomial-time algorithm achieving an \(\f(1)\)-approximation to the online/philosopher benchmark.
        Furthermore,
        there is a \good function with \(\f(1) = 2/3\).
\end{restatable}

There are simpler \good functions achieving worse approximation factors, like \(1/2\) and \((1 - 1/e)\).
Moreover, computational results indicate there is a \good function achieving a factor of \(0.698\).

The main conceptual contribution is a new \emph{capped online contention resolution scheme} (\car) for a single divisible item. Unlike standard OCRS \cite{FSZ-SODA16}, the scheme must satisfy two simultaneous requirements: it preserves a constant fraction of every target allocation in expectation, while never allocating more than the target itself. We prove that a $(1-1/e)$ guarantee is achievable using a valuation-independent online rule based on the remaining supply. \snote{more on technique}

To leverage our \car for divisible welfare maximization, we first show that concavity collapses the natural configuration relaxation for the philosopher benchmark, including its online-feasibility constraints, to a polynomial-size concave program. We then prove that DR-submodularity reduces multi-item rounding to independent coordinate-wise \car instances. Combining these ingredients yields a $(1-1/e)$-approximation to the optimal online policy.
We also show in \Cref{sec:hardness} that this philosopher benchmark is computationally nontrivial: computing the optimal online policy exactly is \#P-hard even for a single divisible item.

\medskip
\noindent\emph{Why divisibility matters:} 
A natural approach is to discretize each divisible item into many indivisible copies and invoke known prophet or philosopher inequalities for submodular valuations. Our results show that such a black-box reduction fundamentally loses the structure of divisible resources.
Consider an instance of offline welfare maximization, i.e., all players have a single function in their support, with indivisible items.
The philosopher benchmark is the optimum allocation, since it has unbounded computational power.
Due to the result of Mirrokni, Schapira, and Vondrak \cite{MirrokniSV08a}, no algorithm can achieve a better than \((1 - 1/e)\)-approximation without an exponential number of value oracle queries. 
Hence, for the philosopher benchmark, discretization is limited to \((1- 1/e)\).
In contrast, divisibility allows the configuration relaxation to collapse to a low-dimensional concave program and
we are able to achieve a strictly better than \((1 - 1/e)\)-approximation via \car.

Our second result is a tight prophet inequality via simple posted prices.


\begin{restatable}{theorem}{postedpricecomp}
\label{teo:postedpricecomp}
    For the online welfare with divisible items (\OWDI) problem with monotone concave DR-submodular valuations, there exist efficiently computable linear per-unit item prices such that the fixed-price auction achieves a \(1/2\)-approximation to the offline/prophet benchmark.
\end{restatable}


A key technical ingredient is a continuous analogue of \emph{supporting prices} for XOS/submodular set functions.
For an allocation vector $\bx$, we define prices using a weighted average of the \emph{Aumann--Shapley value} (a path integral of the gradient along the ray from $0$ to $\bx$) \cite{AumannShapley-Book75}. 
We show that for monotone DR-submodular valuations these prices support every coordinatewise dominated allocation $0\le \by\le \bx$ in the same sense as supporting prices for sets.
This lets us transfer the standard balanced-prices analysis of prophet inequalities to the divisible setting \cite{DuttingFKL20a,FeldmanGL15a}.
The main additional challenge is aggregation: Aumann--Shapley supporting prices are naturally defined separately for each player and realized type, while a posted-price mechanism must announce a single per-unit price for each item. We show that the appropriate welfare-weighted aggregation preserves the revenue--utility decomposition required by the balanced-prices proof. 
The resulting prices are computed from an optimal solution of the low-dimensional offline relaxation.

Divisibility is also crucial for this result.
Discretization can recover the factor $1/2$ using \cite{FeldmanGL15a}, but naturally produces copy-dependent supporting prices because a divisible item can be split among multiple buyers. Our analysis instead yields a single linear per-unit price for every original divisible resource.

}

\subsection{Related Work}
\emph{Philosopher inequalities:}
Philosopher inequalities benchmark against the \emph{optimal online policy}
rather than the offline optimum. This benchmark has been studied extensively
for online Bayesian matching and related indivisible-item allocation problems,
leading to stronger approximation guarantees and new rounding/selection
primitives
\cite{PapadimitriouPSW21a,TorricoT22a,BravermanDL22a,BKPS-EC24,BravermanDPSW25a,NSW-SODA25,STW-SODA26}.
Since discretizing a DR-submodular function gives a submodular function, any results obtained for submodular philosopher inequality in the indivisible setting can be applied to our problem.
Thus, a \((1/2 + \delta)\)-approximation, for a small constant \(\delta>0\), follows from \cite{STW-SODA26}.
However, as stated before, this approach cannot exceed \(1-\nicefrac{1}{e}\) due to hardness of offline submodular welfare maximization~\cite{MirrokniSV08a}.
Our work gives philosopher guarantees for divisible-item welfare with general
concave diminishing-returns valuations. Our hardness result also shows that the optimal online policy can be computationally intractable even when the offline optimum is easy.

\medskip
\noindent \emph{Prophet inequalities and pricing:}
Prophet inequalities compare online algorithms to the offline optimum; see
\cite{Lucier17}, \cite[Chapter~30]{EIV-Book23}, and the forthcoming book
\cite{CDFKS-Book26}. In multi-item welfare settings, posted prices are
central because of their simplicity and their connection to
combinatorial auctions. Feldman, Gravin, and Lucier \cite{FeldmanGL15a} show
that fixed prices achieve a $1/2$-approximation for XOS, and hence submodular,
valuations with indivisible items. D\"utting et al.\ \cite{DuttingFKL20a}
develop the balanced-prices framework, which we adapt using continuous
supporting prices derived from Aumann--Shapley values. Recent work has also
studied subadditive valuations with indivisible items
\cite{CC-STOC23,DKL-FOCS20}, though obtaining an $\Omega(1)$-approximation via
posted prices remains open.

\medskip
\noindent
\emph{Divisible resources and concave returns:}
A large literature studies divisible resource allocation with concave
objectives when the \emph{items} or impressions arrive online, including online
matching and AdWords \cite{KarpVV90a,MehtaSVV07a}, and extensions to concave
returns \cite{FKMMP-WINE09,Patton-SODA26,HJPSZ-FOCS24,DJ-STOC12}. These models
are orthogonal to \OWDI: the online dimension is the arrival of items rather
than buyers with multi-dimensional valuations. Consequently, the algorithmic
focus in that line is typically on improving greedy-type $1/2$ guarantees
toward $(1-1/e)$.


\medskip
\noindent
\emph{Continuous DR-submodularity:} Continuous DR-submodular functions have been studied extensively in offline optimization, including structural characterizations and approximation algorithms \cite{BianLK-NeurIPS17,BianBK20a,NRW-JMLR20}. Our use is different: we exploit diminishing returns in an online welfare problem, both to define continuous supporting prices and to reduce multi-item rounding to our per-item capped OCRS.

\section{Model and Concave Relaxations}
\label{sec:prelim}
We formally define our model and present concave relaxations to our two benchmarks.

\subsection{Model}

An instance of the \emph{Online Welfare with Divisible Items} (\OWDI) problem consists of \(m\) items and \(n\) players.
For each item \(j \in [m]\), there is a unit supply.
For each player \(i \in [n]\), a finite support distribution \(\Dcal_i\) over nonnegative functions with domain \([0, 1]^m\) is given.
Then, one by one, the players arrive:
when player \(i\) arrives, they sample a valuation function \(V_i\) from
\(\Dcal_i\),  and we must immediately decide
what fraction $x_{i,j}$ of each item $j$ will be assigned to \(i\), while ensuring that $\sum_{i=1}^n x_{i,j}\leq 1$.
The goal is to maximize in expectation the sum of player valuations, $\Ebb[\sum_{i=1}^n V_i(\bx_i)]$, 
where $\bx_i := (x_{i,j})_{j\in [m]}$ denotes the {random} $m$-dimensional vector of allocations to player $i$ {that can depend on the observed random valuation functions $V_1,\ldots, V_i$.}

We write \(\Dcal \coloneqq \Dcal_1\times\cdots\times \Dcal_n\)
for the product distribution over valuation profiles, and use
\(\vc{V}\sim \Dcal\) to denote a random valuation profile.

\medskip
\noindent\textbf{Benchmarks.}
Given an instance of \OWDI, we will be interested in the following two benchmarks:
\begin{enumerate}
    \item \emph{Prophet/Offline Benchmark}: It is the expected value of the hindsight optimum solution that knows all the realized valuations $(V_i)_{i\in [n]}$ and then finds an allocation that maximizes $\sum_{i=1}^n V_i(\bx_i)$ while satisfying unit capacity item constraints.
    
    \item \emph{Philosopher/Online Benchmark}: It is the expected value of the optimal online policy that has the same information as our online algorithm but is computationally unbounded (hence, it can solve the exponentially/infinitely sized underlying MDP).
\end{enumerate}

Observe that the online benchmark is a tighter benchmark than the offline benchmark. Nonetheless, we also consider the offline benchmark and compare it against simpler pricing-based algorithms.

\medskip
\noindent\textbf{Value distributions.}
In this work, we assume that each valuation distribution \(\Dcal_i\) has finite support and is supported on normalized, monotone, concave DR-submodular functions \(v:[0,1]^m\to \mathbb R_+\), where normalized means \(v(\mathbf 0)=0\).

\begin{restatable}[DR-Submodularity \cite{BianLK-NeurIPS17,BianBK20a}]{proposition}{DRbasic}
\label{pro:DRbasic}
    For twice continuously differentiable function $f:\Reals_+^m\rightarrow \Reals_+$, the following definitions of DR-submodular functions are equivalent:
    \begin{enumerate}
        \item $f(\vc{y} + \vc{z}) - f(\vc{y}) \geq f(\vc{x} + \vc{z}) - f(\vc{x})$
     for \(\vc{x}, \vc{y}, \vc{z} \in \Reals^m_+\) with \(\vc{y} \leq \vc{x}\). 
    \item For any \(\vc{x},\vc{y} \in \Reals^m_+\) such that \(\vc{y} \leq \vc{x}\), it holds that \(\;\nabla f(\vc{y}) \geq \nabla f (\vc{x})\).
    \item $\frac{\partial^2}{\partial x_i \partial x_j} f(\vc{x}) \leq 0$  for all \(\vc{x} \in \Reals^m_+\) and \(i, j \in [m]\).
    \end{enumerate}
\end{restatable}
Thus DR-submodularity is distinct from concavity: DR-submodularity is an
entrywise condition on the Hessian, whereas concavity requires the Hessian to be
negative semidefinite. We assume both properties. 
This class contains many standard divisible-resource valuations, including
budget-additive valuations \(v(x)=\min\{\sum_j b_jx_j,B\}\), separable concave
valuations, and more general polymatroidal concave returns \cite{Patton-SODA26}. We defer the proof of \Cref{pro:DRbasic} to \Cref{sec:appendixDR}. 


\medskip
\noindent\textbf{Computational model.}
Each distribution \(\Dcal_i\) is specified by its finite support
\(\{v_{i,1},\ldots,v_{i,d_i}\}\) and probabilities
\(\{p_{i,1},\ldots,p_{i,d_i}\}\). For our polynomial-time guarantees, we assume
standard evaluation access sufficient to optimize concave functions
over boxes and approximate the Aumann--Shapley line integrals used in our
prices. Concretely, we assume the valuations are continuously differentiable
and \(L\)-smooth; nonsmooth examples can be handled by an arbitrarily small
smoothing.

\subsection{Concave Relaxations}
\label{subsec:concaverelax}

We next introduce the relaxations that will serve as benchmarks for our
algorithms. We begin with a natural divisible analogue of the configuration LP
used in philosopher inequalities for indivisible items~\cite{STW-SODA26}. This
program is a relaxation of the optimal online policy. We then use concavity to
collapse this infinite-dimensional relaxation to a low-dimensional concave
program. Finally, by weakening the online feasibility constraint, we obtain a
relaxation of the offline/prophet benchmark.


Consider the following infinite configuration relaxation. Formally, the
variables below may be viewed as finite measures over \([0,1]^m\); for
readability we write them as sums over bundles \(\bz\in[0,1]^m\).

\begin{restatable}{lpbox}{configrelax}
\begin{equation}
\label{eq:confLP}
\tag{Configuration LP}
\begin{aligned}
&\max \;  \textstyle \sum_{i = 1}^n \sum_{k = 1}^{d_i} \sum_{\bz \in [0,1]^m} \; y_{i, k, \bz} v_{i, k}(\bz) \qquad &\textnormal{subject to}\\
\end{aligned}
\end{equation}
\begin{equation*}
\begin{aligned}
&\textstyle 
\sum_{\bz \in [0,1]^m} y_{i, k, \bz} \leq p_{i,k} &&\forall i \in [n], \forall  k \in [d_i],\\
&\textstyle \sum_{\bz\in [0,1]^m} z_j y_{i, k, \bz} \leq p_{i, k}
\left(
1 - \sum_{\ell < i}
\sum_{r = 1}^{d_{\ell}}\sum_{\bz \in [0,1]^m} z_j y_{\ell, r, \bz}
\right)
&&\forall i \in [n], \; \forall k \in [d_i],  \;  \forall j \in [m] \\
&y_{i, k, \bz} \geq 0 &&\forall i \in [n],\; \forall k \in [d_i], \;\forall \bz \in [0,1]^m.
\end{aligned}
\end{equation*}
\end{restatable}
The variable \(y_{i,k,\bz}\) denotes the unconditional probability with
which player \(i\) has type \(v_{i,k}\) and receives bundle \(\bz\). 
The first constraint says that summing over all bundles will give the probability of player \(i\) sampling \(v_{i, k}\).
The second constraint is the online feasibility constraint, which says that the valuation sampled by \(i\) and the bundles assigned to previous players are independent for any online algorithm.
Similar online constraints have appeared in previous works on philosopher's inequalities with indivisible items \cite{TorricoT22a,PapadimitriouPSW21a,BravermanDPSW25a,BKPS-EC24}.
See \Cref{sec:appendixLP} for a formal proof that \ref{eq:confLP} is a relaxation of the online benchmark.

If the valuations \(v_{i, k}(\cdot)\) are concave, we are able to write a concave relaxation equivalent to the Configuration LP  in the low-dimensional space of assignment variables \(x_{i, k, j}\).

\begin{restatable}{lemma}{offlinerelax}
\label{lem:concave}
    If the valuations are concave,
    the mathematical program \eqref{eq:onLPm} has the same objective value as the \ref{eq:confLP}.
\begin{lpbox}
\begin{equation}
\label{eq:onLPm}
\begin{aligned}
\max \; & \textstyle \sum_{i = 1}^n \sum_{k = 1}^{d_i} \; p_{i,k} v_{i,k}(\vc{x_{i,k}}/p_{i, k}) &\textnormal{subject to}\\
&\textstyle x_{i, k, j} \leq p_{i, k} \left(1 - \sum_{\ell < i} \sum_{r = 1}^{d_{\ell}} x_{\ell,r, j}\right)
&&\forall i \in [n], \; \forall k \in [d_i], \; \forall j \in [m],\\
&x_{i, k, j} \geq 0 &&\forall i \in [n],\; \forall k \in [d_i], \;\forall j \in [m].\\
\end{aligned}
\end{equation}
\end{lpbox}
\end{restatable}

The complete proof of \Cref{lem:concave} can be found in \Cref{sec:appendixLP}.
The proof constructs a mapping between feasible solutions of the two programs. Given a feasible \(y\) we can construct \(x\) as
$x_{i, k, j} = \sum_{\bz \in [0, 1]^m} z_j y_{i, k, \bz}$.
Thus \(x_{i,k,j}\) is the unconditional expected amount of item \(j\)
allocated to player \(i\) on type \(k\).
Hence \(\vc{x}_{i,k}/p_{i,k}\) is the conditional expected bundle for
player \(i\) given type \(k\).
The equivalence crucially relies on Jensen's inequality:
the expected value of the bundle given by the variables \(y_{i,k,\bz}\) is upper bounded by the value of the expected bundle 
\(\vc{x_{i, k}}/p_{i,k}\).


For the offline benchmark, we do not have independence between the sampling probability and the allocation to previous players.
Hence, our relaxation changes the first constraint in \eqref{eq:onLPm} with the following:
\begin{equation}
\label{eq:offLPm}
\begin{aligned}
&\textstyle x_{i, k, j} \leq p_{i, k}
&&\forall i \in [n], \; \forall k \in [d_i], \; \forall j \in [m],\\
&\textstyle \sum_{i = 1}^n \sum_{k = 1}^{d_i} x_{i, k, j} \leq 1 &&\forall j \in [m].
\end{aligned}
\end{equation}

\section{Philosopher Inequalities}
\label{sec:phi}

In this section, we prove our main philosopher inequality by rounding any
feasible solution of the online concave relaxation \eqref{eq:onLPm}. Concavity was used in \Cref{subsec:concaverelax} to obtain this low-dimensional relaxation. In this section we will use only monotonicity and DR-submodularity to design a rounding algorithm. To state our general rounding result, we begin by introducing \good functions.


\begin{definition}[\good function]
\label{def:good}
A concave function \(\f: \Reals_+ \to \Reals_+\)
is \good if
\begin{itemize}
    \item \(\f(z) \leq \min\{z, 1\} \qquad \forall z \in \Reals_+\),
    \item \((1 - \mu) \cdot \f\left(\frac{z - \f(z)}{\mu}\right) + \mu \cdot \f \left(\frac{z}{\mu}\right) \geq \f(z) \qquad \forall z \in \Reals_+, \mu \in (0,1]\).
\end{itemize}
\end{definition}

We will later prove that these conditions naturally arise from our rounding proof.
The following is our main result of this section.

\begin{theorem}
\label{lem:phirounding}
    For an instance of the \OWDI problem with monotone DR-submodular valuations, let  \(\vc{x}\) be any feasible solution to \eqref{eq:onLPm} and let \(\f\) be any \good function. There is an online policy whose expected welfare is at least \(\f(1)\cdot \sum_{i = 1}^n \sum_{k = 1}^{d_i} p_{i,k} v_{i, k} (\vc{x_{i,k}}/p_{i,k})\).
\end{theorem}

In \Cref{sec:appendixgoodf}, we will show that there exists a \good function with $\f(1) = \frac23$. Thus,  \Cref{teo:mainphi} immediately follows when we apply  \Cref{lem:phirounding} to the optimal solution of the relaxation~\eqref{eq:onLPm}, which upper-bounds the philosopher benchmark.

The proof of \Cref{lem:phirounding} proceeds in two steps:
\begin{enumerate}
\item We define the single-item \car problem and prove a
      $\f(1)$-\car, given a \good function \(\f\). This appears in \Cref{sec:car} and \Cref{sec:martingale}.

\item We prove a decomposition lemma to show that running this \car separately for every item preserves each player's multi-item value by DR-submodularity. This appears in \Cref{subsec:reduction}.
\end{enumerate}

\subsection{The \car Problem}\label{sec:car}

We formally define the single item \car problem for some parameter $\alpha\in [0,1]$. 

\begin{definition}[$\alpha$-\car]
\label{def:car}
    Given an \OWDI instance with a single item and a feasible solution
    \(x_{i, k}\) to the concave relaxation of the online benchmark \eqref{eq:onLPm}, we say that an online algorithm is an $\alpha$-\car if
    it holds that for every \(i \in [n]\) and for every \(k \in [d_i]\), conditioned on the event $V_i=v_{i,k}$, the following two conditions are satisfied:
    \begin{enumerate}
        \item \emph{$\alpha$-Selectability}: Player \(i\) receives in expectation at least \(\alpha \cdot  x_{i, k}/p_{i, k}\) fraction of the item.

        \item \emph{Capping}: Player \(i\) never receives more than \(x_{i, k}/p_{i, k}\) of the item.
    \end{enumerate}
\end{definition}

Intuitively, the first condition is the standard selectability condition for OCRS \cite{FSZ-SODA16}: since the probability that $V_i=v_{i,k}$ equals $p_{i, k}$, this condition says that on average player $i$ receives $\alpha \cdot  x_{i, k}$ fraction of the item. The second property is needed because our online rounding algorithm is valuation independent, which means that it must work even for concave functions that are identically zero beyond the fractional solution.


Fix a one-item instance and a feasible solution \(\vc{x}\) to the single-item relaxation. Since there is only one item, we suppress the item index. Define the fractional remaining supply before player \(i\) by
\[ \textstyle \rem_i \coloneqq 1 - \sum_{\ell < i} \sum_{r = 1}^{d_\ell} x_{\ell, r}.\] 
Thus \(\rem_i\) is the amount of the item remaining before player \(i\) in the fractional solution. The online constraint in \eqref{eq:onLPm} can now be compactly written as
\[
x_{i, k} \leq p_{i,k}\, \rem_i \qquad \text{for every \(i \in [n]\) and \(k \in [d_i]\)}.
\]

In \Cref{algo:mainphilo}, we describe a general construction for a \car, depending on a given \good function \(\f: \Reals_+ \to \Reals_+\).

\begin{algorithm}
\caption{General-\car ($\vc{x},\f$)}
\label{algo:mainphilo}
\DontPrintSemicolon

Start with $R_1 \gets 1$.

When each player $i \in [n]$ arrives: 
\begin{enumerate}
    \item \begin{enumerate}
        \item If \(\rem_i = 0\), set \(Z_i = 0\)
        \item If \(\rem_i > 0\) and the $i$-th player's valuation is $V_i = v_{i,k}$, then
        set \(Z_i \gets \frac{x_{i, k}}{p_{i, k}} \cdot \f\left(R_i/\rem_i\right)\) \,.
    \end{enumerate}
    \item Assign \(Z_i\) to player \(i\) and update \(R_{i+1} \gets R_i - Z_i\) \,.
\end{enumerate}
\end{algorithm}

\begin{lemma}
Let \(\f: \Reals_+ \to \Reals_+\) be a \good function.
Then \Cref{algo:mainphilo} is a \(\f(1)\)-\car.
\end{lemma}

\begin{proof}

First observe that the algorithm never assigns to a player more than
\(x_{i, k}/p_{i,k}\) because \(\f\) is less than or equal to \(1\) at every point.
Moreover note that \Cref{algo:mainphilo} will never give to a player more than what is remaining of the item because
\[
Z_i ~=~ \frac{x_{i, k}}{p_{i, k}}\f(R_i /\rem_i) ~\leq~ \frac{x_{i, k}}{p_{i, k}\rem_i} R_i
~\leq~ R_i,
\]
where the first inequality uses that \(\f(z) \leq z\) and the last inequality follows from the online constraint $x_{i,k} \le p_{i,k}\rem_i$.

In order to prove the \(\f(1)\)-selectability property, we need the following submartingale lemma, which is the heart of the proof.

\begin{restatable}[Submartingale Lemma]{lemma}{decphi}
\label{lema:decphi}
    Define the potential function 
    \[
    \phi(i) \coloneqq
    \begin{cases}
    \Ebb[\f(R_i/\rem_i)], &\text{if \(\rem_i > 0\)}, \\
    +\infty, &\text{if \(\rem_i = 0\)}.
    \end{cases}
    \]
    Then \(\phi\) is increasing in $i$. In particular, this implies  \(\phi(i) \geq \phi(1) = \f(1)\).
\end{restatable}
Before proving the lemma, we show why it implies that the first condition in \Cref{def:car} holds for \(\alpha = \f(1)\).

Let \(i \in [n]\) and \(k \in [d_i]\). We may assume that \(\rem_i > 0\),  since  \(\rem_i = 0\) implies \(x_{i, k} = 0\), where
the statement is trivially true. 
By linearity of expectation,
\[
     \Ebb[Z_i \mid V_i = v_{i,k}] = \frac{x_{i,k}}{p_{i,k}}\left(\Ebb[\f(R_i/\rem_i ) \mid V_i = v_{i,k}]\right).
\]
However, since the algorithm is online, the remaining \(R_i\) is independent of the realization of \(i\), so
\[ \Ebb[Z_i \mid V_i = v_{i,k}] = \frac{x_{i,k}}{p_{i,k}}\phi(i) .\]
By \Cref{lema:decphi}, we know \(\phi(i) \geq \phi(1) = \f(1)\).
Consequently, \(\Ebb[Z_i \mid V_i = v_{i,k}] \geq
\f(1) \frac{x_{i,k}}{p_{i,k}} \).    
\end{proof}

\subsection{Proof of Submartingale \texorpdfstring{\Cref{lema:decphi}}{}}\label{sec:martingale}

    Let \(i \in [n - 1]\) and let \(d \coloneqq d_i\) denote the support size of \(\Dcal_i\).
    If \(\rem_{i + 1} = 0\), then \(\phi(i + 1) = +\infty \geq \phi(i)\) and we are done. If \(\rem_i = 0\), then \(\rem_{i + 1} = 0\) since \(\rem_i\) is nonincreasing in \(i\), and the proof is complete again. From now on, we will assume that \(\rem_i\) and \(\rem_{i + 1}\) are strictly positive.

    Fix \(r \in [0, 1]\). If we can prove that
    \[\Ebb[\f(R_{i + 1}/\rem_{i + 1}) \mid R_i = r] \geq \f(r/\rem_i)\]
    where the expectation is over \(\Dcal_i\),
    then taking the expectation over \(R_i\) on both sides yields \(\phi(i + 1) \geq \phi(i)\).
    Expanding the expectation, we have
    \[
        \sum_{k = 1}^d p_{i, k}\;
        \f\left(
            \frac{r - \frac{x_{i, k}}{p_{i, k}}\f(r/\lambda_i) }{\lambda_{i + 1}}
        \right)
        \geq
        \f(r/\lambda_{i}).
    \]
    Define the function \(g: y \in [0, \lambda_i] \mapsto \f\left(\frac{r-y\f(r/
    \lambda_i)}{\lambda_{i+1}}\right)\) and observe that since \(\f\) is concave and \(g\) is \(\f\) composed with an affine transformation, \(g\) is also concave.
    Observe that by the online constraint,
    \(x_{i,k}/p_{i, k} \leq \lambda_i\), thus the concavity of \(g\) allows us to conclude that
    \[
    g\left(\frac{x_{i,k}}{p_{i,k}}\right)
    \geq
    \frac{x_{i,k}}{p_{i,k} \lambda_i} g(\lambda_i)
    +
    \left(1-\frac{x_{i,k}}{p_{i,k} \lambda_i} \right) g(0).
    \]

    Hence, it holds that
    \[
    \sum_{k = 1}^d p_{i, k}
        \;\f\left(
            \frac{r - \frac{x_{i, k}}{p_{i, k}}\f(r/\lambda_i) }{\lambda_{i + 1}}
        \right)
    =   
    \sum_{k = 1}^d p_{i, k}
    g\left(\frac{x_{i,k}}{p_{i,k}}\right)
    \geq
    \sum_{k=1}^d \frac{x_{i,k}}{\lambda_i} g(\lambda_i)
    +
    \left(1 - \sum_{k=1}^d \frac{x_{i,k}}{\lambda_i}\right) g(0).
    \]
    By noticing that \(\sum_{k = 1}^d x_{i, k} = \lambda_i - \lambda_{i + 1}\) and expanding the definition of \(g\), it is sufficient to prove that
    \[
    \left(1 - \frac{\lambda_{i+1}}{\lambda_i}\right)
    \f\left(\frac{r-\lambda_i \f(r/
    \lambda_i)}{\lambda_{i+1}}\right)
    +
    \frac{\lambda_{i+1}}{\lambda_i}
    \f\left(\frac{r}{\lambda_{i+1}}\right)
    \geq
    \f\left(\frac{r}{\lambda_i}\right).
    \]
    Define \(z \coloneqq r/\lambda_i\) and \(\mu \coloneqq \lambda_{i+1}/\lambda_i\).
    Since we assumed \(\lambda_i, \lambda_{i + 1} > 0\) and by the fact the \(\lambda\)'s are non-increasing in \(i\), we have that
    \(z \in \Reals_+\) and \(\mu \in (0,1]\).
    Furthermore, the inequality becomes
    \[
    \left(1 - \mu\right)
    \f\left(\frac{z \lambda_i-\lambda_i \f(z)}{\lambda_{i+1}}\right)
    +
    \mu\;
    \f\left(\frac{z \lambda_i}{\lambda_{i+1}}\right)
    \geq
    \f\left(z\right).
    \]
    Dividing the numerator and denominator of the fractions inside \(\f\) by \(\lambda_i\)
    leads precisely to the second requirement for \good functions.
   
\subsection{Reducing Multi-Item Rounding to Single-Item \car}
\label{subsec:reduction}

We now show that single-item \car can be run separately
on each item and composed using DR-submodularity. The only subtlety is that the single-item rounding decisions for different items may be correlated through the same realized type of a player.




\begin{lemma}[Decomposition Lemma]
\label{lema:DRavg}
    Consider any fixed vector \(\vc{x} \in \Reals^m_+\) and any fixed scalar \(\alpha \in [0,1]\).
    Let  \(\Dcal\) be a distribution over $\Reals^m_+$ such that every vector $\vc{y}$ in \(\supp(\Dcal)\) 
    satisfies \(\vc{y} \leq \vc{x}\) and additionally,
    \(\Ebb_{\Dcal}[\vc{y}] \geq \alpha \vc{x}\). Then, for any  monotone nonnegative DR-submodular function \(f:[0,1]^m \rightarrow \Reals_+\):
    \[\Ebb_{\Dcal}[f(\vc{y})] \geq \alpha \cdot f(\vc{x}).\]
\end{lemma}
\begin{proof}
    For any vector \(\vc{y} \in \Reals^m\) and \(j \in [m]\), define 
    \(\vc{y}_{\leq j} \coloneqq (y_1, \dots, y_j, 0, \dots, 0)\).
    Moreover, note that for any \(\vc{z} \in \Reals^m_+\) and \(j \in [m]\),
    the function \(g(t) = f(\vc{z} + te_j)\) is concave in \(t\) since
    \[ 
    \textstyle
        g''(t) ~=~ \frac{\partial^2}{\partial t^2} f(\vc{z} + te_j) ~=~ (\nabla^2 f(\vc{z} + te_j))_{jj} ~\leq~ 0.
    \]

    Recall that for any one-dimensional concave function \(g\) with $g(0) \ge 0$, it holds that \(g(a) - g(0) \geq \tfrac{a}{b}(g(b) - g(0))\) for \(0 < a \leq b\).
    Hence, for any vector \(\vc{y} \in \supp(\Dcal)\) and any \(j \in [m]\), we have
    \[
    f(\vc{y}_{\leq j}) - f(\vc{y}_{\leq j-1}) ~\geq~ \frac{y_j}{x_j} \Big( f(\vc{y}_{\leq j-1} + x_j e_j) - f(\vc{y}_{\leq j-1}) \Big) ~\geq~ \frac{y_j}{x_j}\Big(f(\vc{x}_{\leq j}) - f(\vc{x}_{\leq j-1}) \Big),
    \]
    where the first inequality uses concavity along \(e_j\) and
    the second inequality uses DR-submodularity with \(\vc{z} = (\vc{x} - \vc{y})_{\leq j-1}\) (noting $\vc{y} \le \vc{x}$, so marginal gains are higher at $\vc{y}$).
    Taking expectation with respect to \(\Dcal\) and using the monotonicity of \(f\), we obtain
    \[
    \Ebb_\Dcal[f(\vc{y}_{\leq j})] - \Ebb_\Dcal[f(\vc{y}_{\leq j-1})]
    ~\geq~
    \frac{\Ebb_\Dcal[y_j]}{x_j}\Big(f(\vc{x}_{\leq j}) - f(\vc{x}_{\leq j-1})\Big)
    ~\geq~
        \alpha\cdot \Big(f(\vc{x}_{\leq j}) - f(\vc{x}_{\leq j-1}) \Big).
    \]
    Note that if \(x_j = 0\), then \(y_j = 0\) for every \(y\) in \(\supp (\Dcal)\) and the previous inequality is trivially true.
    Summing over all \(j \in [m]\) and using that \(f(\vc{0})\geq 0\), we obtain
    \[
    \Ebb_\Dcal[f(\vc{y})] ~\geq~ \alpha \cdot f(\vc{x}). \qedhere
    \]
\end{proof}

The final step is to use \Cref{lema:DRavg} to connect the single item \(\alpha\)-\car with online rounding for \OWDI.

\begin{lemma}
\label{lema:composition}
    Let \(\ALG\) be a single item $\alpha$-\car for some \(\alpha \in [0,1]\). 
    Then there is an algorithm \(\widehat{\ALG}\) that
    achieves an \(\alpha\)-approximation for \OWDI with monotone concave DR-submodular functions over multiple items.
    Moreover, if  \(\ALG\) runs in polynomial time then \(\widehat{\ALG}\) also runs in polynomial time.
\end{lemma}

\begin{proof}
    Assume we have an instance of \OWDI and a feasible solution \(\vc{x}\)
    for the concave relaxation of the online benchmark \eqref{eq:onLPm}.
    For each item \(j \in [m]\), define \(\ALG_j\) to be the policy that mimics
    \(\ALG\) as if \(j\) was the only item in the instance and the variables
    \(\vc{x}\) were restricted to \((x_{i, k, j})_{i,k}\). Observe that the restricted fractional solution satisfies the strengthened convex program formulated for the single item $j$.

    Let \(\widehat{\ALG}\) be the algorithm that separately runs \(\ALG_j\)  for each item $j\in [m]$. Note that different  \(\ALG_j\) could be correlated since player $i$ takes the same valuation $v_{i,k}(\cdot)$ in each of them. 
    By hypothesis, for every \(j \in [m]\), we have
    \[
    \Ebb[\text{player $i$ receives of item \(j\)} \mid V_i = v_{i,k}]
    \geq \alpha \cdot {x_{i, k, j}}/{p_{i, k}},
    \]
    and player \(i\) never receives more than \(x_{i, k, j}/p_{i, k}\) from item \(j\).
    Hence, combining this for all \(j\), we have that, conditioned on the realization of the function \(v_{i, k}(\cdot)\),
    player \(i\) receives in expectation at least the vector \(\alpha \vc{x_{i, k}}/p_{i, k}\) and the amount
    received is less than or equal to \(\vc{x_{i, k}}/p_{i,k}\) coordinate-wise.

    Let \(\widehat{\ALG}_i\) be the welfare \(\widehat{\ALG}\) obtains from player \(i\).
    By \Cref{lema:DRavg}, it holds that
    \[
        \Ebb[\widehat{\ALG}_i \mid  V_i = v_{i, k}] \geq \alpha \cdot v_{i, k}(\vc{x_{i, k}}/p_{i, k}).
    \]
    Consequently, $\Ebb[\widehat{\ALG}]$ equals
    \[ \textstyle
    \sum_{i = 1}^n \Ebb[\widehat{\ALG}_i] ~=~
    \sum_{i = 1}^n \sum_{k = 1}^{d_i} \Ebb[\widehat{\ALG}_i \mid V_i = v_{i, k}]
    \Pr[V_i = v_{i,k}]
    ~\geq~
    \alpha \sum_{i = 1}^n \sum_{k = 1}^{d_i} \; p_{i, k} v_{i, k}(\vc{x_{i,k}}/p_{i,k}).
    \]
    Note that the RHS is precisely \(\alpha\) times the objective value of the solution \(\vc{x}\) in \eqref{eq:onLPm}.
\end{proof}

\section{Prophet Inequalities}
\label{sec:pro}

In this section, we prove that a fixed-price auction (\Cref{algo:postedprice}) achieves a \(1/2\)-approximation to the offline/prophet benchmark.

\postedpricecomp*

We remark that under our assumptions from \Cref{sec:prelim}, the demand oracle needed for \Cref{algo:postedprice} can be implemented efficiently, because  utility maximization can be solved as a convex optimization problem, by maximizing a concave objective over a convex set. 
 Thus our entire algorithm can be implemented to run in polynomial time.

The approximation factor $1/2$ in \Cref{teo:postedpricecomp} is tight even for the case of a single item and linear valuations (see \Cref{sec:hardness}). If efficient computation is not a concern, we prove in \Cref{sec:missingProphet} that the theorem holds even if we relax concavity—i.e., a $1/2$-approximation fixed-price auction exists when the valuations are only monotone DR-submodular—by working with the expected hindsight optimal solution instead of the concave relaxation.


\begin{algorithm}[!ht]
\caption{Fixed-Price-Auction$(\vc{p})$}
\label{algo:postedprice}
\DontPrintSemicolon

     Start with $\vc{R}_1\in \Reals^m$ being the all-ones vector (unit supply).
     
     When player $i \in [n]$ arrives: 
        \begin{enumerate}
            \item Let $\vc{Z}_i \gets 
        \arg\max \{ v_i(\vc{x}) - \vc{p}^\transp \vc{x} :
        \vc{0} \le \vc{x} \le \vc{R}_i \}$
        \item Assign $\vc{Z}_i$ to player $i$ and set
        $\vc{R}_{i+1} \gets \vc{R}_i - \vc{Z}_i$\;
        \end{enumerate}

\end{algorithm}

\subsection{Supporting Prices via Aumann–Shapley Values}

Our main contribution in the proof of \Cref{teo:postedpricecomp} is a continuous analogue of ``supporting prices'', akin to supporting prices defined for XOS valuations, with indivisible items \cite{dobzinski2007two,FeldmanGL15a}.

We start by recalling the definition of supporting prices in the setting of indivisible items.

\begin{definition}[\cite{dobzinski2007two,FeldmanGL15a}]
    A monotone set valuation function $v: 2^U \rightarrow \mathbb{R}_+$ admits \emph{supporting prices} if for any subset $S \subseteq U$ there exists  item prices $\vc{p}_S:U\rightarrow \mathbb{R}_+$ such that
    \begin{enumerate}[label = (\roman*)]
    \item $v(S) = \sum_{j \in S} p_S(j)$ and
    \item for any subset $T\subseteq S$, we have $v(T) \geq \sum_{j\in T} p_S(j)$.
    \end{enumerate}
\end{definition}

For example, when considering an XOS valuation\footnote{Max of additive functions.} \(v\),
the supporting prices for a set \(S\) form a linear function \(a\) such that \(v(S) = a(S)\) and \(v(T) \geq a(T)\) for all \(T \subseteq S\).

The following lemma generalizes the definition of supporting prices to continuous settings, and shows that for DR-submodular functions the ``Aumann–Shapley'' values \cite{AumannShapley-Book75} form supporting prices.

\begin{lemma}
    \label{lema:DRprice}
    Let \(v: [0, 1]^m \to \Reals_+\) be a monotone DR-submodular function with \(v(0) = 0\).
    For any \(\bx \in [0,1]^m\), define the \emph{Aumann–Shapley} vector
    \[
\textstyle        q(\bx) \coloneqq \int_{0}^1 \nabla v(t \bx) \; dt.
    \]
    Then, $q(\bx)$ forms \emph{supporting prices}, i.e., they satisfy
    \begin{enumerate}[label = (\roman*)]
        \item \label{item:DRi} \(v(\bx) = \iprod{q(\bx)}{\bx}\) for every \(\vc{0} \leq \bx\).
        \item \label{item:DRii} \(v(\vc{y}) \geq \iprod{q(\bx)}{\vc{y}}\) for every \(\vc{0} \leq \vc{y} \leq \bx\).
    \end{enumerate}
\end{lemma}
\begin{proof}
The first property follows from the well-known ``efficiency'' property of Aumann–Shapley values for any monotone differentiable function; we present a short self-contained proof below. The second property crucially relies on the DR-submodularity of $v$.

    Define $g:[0,1] \mapsto  \mathbb{R}$ by    
    \(g(t)= v(t \bx)\). By the chain rule, \(g'(t) = \iprod{\nabla v(t \bx)}{\bx}\). It holds that
    \[\textstyle v(\bx) ~=~ g(1) - g(0) ~=~ \int_0^1 g'(t) \; dt ~=~ \iprod{q(\bx)}{\bx}.\]

    For the second property, note that by DR-submodularity \(\nabla v(t\bx) \leq \nabla v(t \vc{y})\) for every \(t \in [0, 1]\) (since \(t\bx \ge t\vc{y}\)); thus, by integrating over $t$, we obtain \(q(\vc{y}) \geq q(\bx)\). Consequently,
    \[
        v(\vc{y}) ~=~ \iprod{q(\vc{y})}{\vc{y}} ~\geq~ \iprod{q(\bx)}{\vc{y}}. \qedhere
    \]
\end{proof}

\paragraph{Connection to Balanced Prices.}
It is also possible to understand the properties of supporting prices from the balanced prices framework introduced in \cite{DuttingFKL20a}. This framework formalizes the intuition that the prices should neither be too high  nor too low. To see this in our context, consider any optimal allocation \(\bx\)
and let \(\vc{y}\) represent a partial allocation such that \(\vc{y} \leq \bx\).
Then, \Cref{item:DRii} gives us that \(v(\bx - \vc{y}) \geq \iprod{q(\bx)}{\bx - \vc{y}}\), which implies the prices are not too high since the function value of the remaining items $\bx - \vc{y}$ is at least the revenue we could obtain from selling them.
Conversely, combining \Cref{item:DRi} and \Cref{item:DRii} gives 
\(\iprod{q(\bx)}{\vc{y}} \geq v(\bx) - v(\bx - \vc{y})\), which shows that
the prices are not too low as they compensate for the loss of valuation by allocating \(y\).

\subsection{Proof of \texorpdfstring{\Cref{teo:postedpricecomp}}{}}
\label{sec:pf1.2}

Let \(x^*_{i, k, j}\) be an optimal solution to the concave relaxation for the offline/prophet benchmark, i.e., relaxation \eqref{eq:onLPm} with the constraints from \eqref{eq:offLPm} instead of the online constraint.

Unlike the indivisible items setting where each item is allocated to only a single player, an item might now be allocated to all $n$ players. Hence, 
we will have a vector of supporting prices for each player \(i \in [n]\) and each possible valuation function \(k \in [d_i]\) defined as
\[ \textstyle
\vc{q}_{i, k} \coloneqq \int_0^1 \nabla v_{i, k}\left(t \frac{\vc{x^*_{i, k}}}{p_{i, k}}\right) \; dt \in \Reals_+^m.
\]
Recall that under our assumptions from \Cref{sec:prelim}, an optimal
solution \(x^*_{i, k, j}\) and the vector \(\vc{q_{i,k}}\) can be computed, up to arbitrary accuracy, in polynomial time.

To eventually define a single fixed price per item $j$, we take the (weighted) average of these supporting prices to define the \emph{welfare contribution} for  item \(j\) as
\[    \textstyle
SW_j \coloneqq \sum_{i = 1}^n \sum_{k = 1}^{d_i} q_{i, k, j} x^*_{i,k,j}.
\]
These welfare contributions satisfy 
\begin{align*}
     \textstyle
    \sum_{j = 1}^m SW_j
    ~=~ \textstyle
    \sum_{j = 1}^m  \sum_{i = 1}^n \sum_{k = 1}^{d_i} q_{i, k, j} x^*_{i,k,j}
    ~&=~ \textstyle
    \sum_{i = 1}^n \sum_{k = 1}^{d_i} p_{i, k} \sum_{j = 1}^m q_{i, k, j} \frac{x^*_{i,k,j}}{p_{i, k}}\\
    ~&=~ \textstyle
    \sum_{i = 1}^n \sum_{k = 1}^{d_i} p_{i, k}
    v_{i,k} \left(\frac{\vc{x^*_{i,k}}}{p_{i, k}}\right),
\end{align*}   
where the last equality uses the first property from \Cref{lema:DRprice}.

\begin{restatable}{lemma}{postedprice}
\label{lem:postedprice}
    For the online welfare with divisible items (\OWDI) problem and monotone concave DR-submodular valuations, \Cref{algo:postedprice} with prices
    \[
    p_j \coloneqq \frac{1}{2} SW_j \qquad \forall j \in [m]
    \]
    is \((1/2)\)-competitive compared to the offline/prophet benchmark.
\end{restatable}

Given the setup, the proof of \Cref{lem:postedprice} follows the standard balanced prices technique from the prophet inequalities literature \cite{FeldmanGL15a,DuttingFKL20a}. 
We present a self-contained proof of a more general statement in \Cref{sec:appendixpro}. The proof proceeds by splitting the welfare into revenue plus utility. The revenue is lower bounded by the sales of items at sufficiently high prices (from the first property of supporting prices in \Cref{lema:DRprice}). To lower bound the utility, we use the fact that players optimize their own utility and could have purchased the remaining fractional items from their relaxation allocation, which they value highly (from the second property of supporting prices in \Cref{lema:DRprice}).
The statement in \Cref{sec:appendixpro} also implies the remark at the beginning of the section that one can prove the existence of a \((1/2)\)-approximation fixed-price auction even when dropping the concavity assumption, however the prices are not efficiently computable.

\subsection{Proof of \texorpdfstring{\Cref{lem:postedprice}}{}}
\label{sec:missingProphet}
\label{sec:appendixpro}

Let \(\{y^\vc{v} \in [0, 1]^{n \times m}: \vc{v} \in \Dcal\}\) be a family of vectors such that
\(y^\vc{v}_{i, j} \leq 1\) for every \(i \in [n], j \in [m], \vc{v \in \Dcal}\) and \( \,\Ebb_{\vc{V}} \big[\sum_{i = 1}^n y^\vc{V}_{i, j}\big] \leq 1\).
For the sake of intuition, one can think of \(\vc{y^\vc{v}}\) as an allocation when profile \(\vc{v}\) is sampled, later we will construct this family from an optimal solution \(\vc{x}^*\).
For each \(i \in [n]\) and \(\vc{v} \in \Dcal\), define the supporting prices
\[
    q^{\vc{v}}_i \coloneqq \int_0^1 \nabla v_i(t \by^{\vc{v}}_i) \; dt \in \Reals_+^m.
\]
The welfare contribution of each item is
    \begin{equation}
        \label{eq:defSW}
        SW^{\vc{v}}_j \coloneqq \sum_{i = 1}^n q^{\vc{v}}_{i, j} y^{\vc{v}}_{ij}.    
    \end{equation}
    And note that
    \begin{equation}
        \label{eq:sumSW}
        \sum_{j = 1}^m SW_j^\vc{v} = \sum_{i = 1}^n \sum_{j = 1}^m q_{i,j}^{\vc{v}} \by_{i, j}^{\vc{v}}
        =
        \sum_{i = 1}^n \iprod{\by^\vc{v}_i}{q^\vc{v}_i} = \sum_{i = 1}^n {v}_i(\by^{\vc{v}}_i),
    \end{equation}
    where in the last equality we used \Cref{lema:DRprice}.
We now prove a statement that will immediately imply \Cref{lem:postedprice}.

\begin{lemma}
\label{lem:generalphi}
    For the online welfare with divisible items (\OWDI) problem and monotone DR-submodular valuations, \Cref{algo:postedprice} with prices
    \[
    p_j \coloneqq \frac{1}{2} \Ebb_{\vc{V}} [SW^{\vc{V}}_j] \qquad \forall j \in [m]
    \]
    achieves expected welfare at least \(\tfrac12 \sum_{i = 1}^n \Ebb_{\vc{V}}[{V}_i(\vc{y}^\vc{V}_i)]\).
\end{lemma}

Note how \Cref{lem:generalphi} does not require concavity of the valuations, however, it is necessary for \Cref{lem:postedprice} since we are using the concave relaxation from \Cref{sec:prelim} to efficiently construct the vectors \(\vc{y}\).
Before proceeding with the proof, we show how to
derive \Cref{lem:postedprice}.

\begin{proof}[Proof of \Cref{lem:postedprice}]
    Let \(x^{*}_{i, k, j}\) be an optimal solution to the offline benchmark relaxation, thus satisfying the constraints in \eqref{eq:offLPm}.
    For every \(i \in [n], j \in [m]\) and \(\vc{v} \in \Dcal\), define
    \[ \textstyle
    y^\vc{v}_{i, j} \coloneqq  \frac{x^*_{i, k, j}}{p_{i, k}} \qquad \text{if \(v_i = v_{i, k}\)}.
    \]
    Note that \(y^\vc{v}_{i, j} \leq 1\) and
    \[ \textstyle
        \Ebb_{\vc{V}}\left[ \sum_{i = 1}^n y^\vc{V}_{i, j}\right]
        =
        \sum_{i = 1}^n \Ebb_{\vc{V}}\left[ y^\vc{V}_{i, j}\right]
        =
        \sum_{i = 1}^n \sum_{k = 1}^{d_i} p_{i, k} \frac{x^*_{i, k, j}}{p_{i, k}}
        =
        \sum_{i = 1}^n \sum_{k = 1}^{d_i} x^*_{i, k, j}
        \leq
        1.
    \]

    Furthermore, the prices defined in \Cref{sec:pf1.2} are the same as the prices defined in this section
    \[ \textstyle
    \Ebb_{\vc{V}}[SW^\vc{V}_j]
    =
    \sum_{i = 1}^n \Ebb_{\vc{V}}[q^\vc{V}_{i, j}y^{\vc{V}}_{i,j}]
    =
    \sum_{i = 1}^n \sum_{k = 1}^{d_i} p_{i, k} q_{i, k, j}
    \frac{x^*_{i, k, j}}{p_{i,k}}
    =
    \sum_{i = 1}^n \sum_{k = 1}^{d_i} q_{i, k, j} x^*_{i, k, j} = SW_j.
    \]
    Thus, by \Cref{lem:generalphi}, \Cref{algo:postedprice} with the vector
    \(\vc{p}\) will have expected value at least
    \[ \textstyle
    \frac{1}{2} \sum_{i = 1}^n \Ebb_{\vc{V}} \left[ {V}_i (\by^\vc{V}_i)\right]
    =
    \frac{1}{2} \sum_{i = 1}^n \sum_{k = 1}^{d_i} p_{i, k} v_{i, k}(\vc{x^*_{i, k}}/p_{i, k}). \qedhere
    \]
\end{proof}

\begin{proof}[Proof of \Cref{lem:generalphi}]

    Consider the execution of \Cref{algo:postedprice} on valuation profile \(\vc{v}\).
    Then, \(\vc{R}^\vc{v}_i \in \Reals^m_+\) is the vector representing how much of each item
    is left before player \(i\) arrives. In particular, \(\vc{R}^\vc{v}_{n + 1}\) is how much is left
    after all players.
    The vector \(\vc{Z}^\vc{v}_i \in \Reals^m_+\) is the fraction of items player \(i\) received.
    The revenue, denoted by \(\rev^\vc{v}\) is the sum of the prices for everything that was sold.
    More precisely,
    \[
        \rev^\vc{v} = \sum_{i = 1}^n \iprod{\vc{p}}{\vc{Z}^\vc{v}_i}.
    \]
    And for each player \(i\), its utility, denoted by \(u^\vc{v}(i)\), is defined as
    \[
        u^\vc{v}(i) = v_i(\vc{Z}^\vc{v}_i) - \iprod{\vc{p}}{\vc{Z}^\vc{v}_i}.
    \]
    Observe that the objective value of the algorithm can be split in revenue plus utility:
    \[
        \sum_{i = 1}^n v_i(\vc{Z}^\vc{v}_i) = \rev^\vc{v} + \sum_{i = 1}^n u^\vc{v}(i).
    \]

    Our goal now is to bound the expected revenue and expected utility for each player.
    To bound the revenue, note that \(\vc{R}^\vc{v}_{n + 1} = \ones - \sum_{i = 1}^n \vc{Z}^\vc{v}_i\) for every \(\vc{v} \in \Dcal\),
    since the amount remaining is the starting amount minus what was assigned. Thus,
    \[
    \Ebb_{\vc{V}}[\rev^\vc{V}] = \sum_{j = 1}^m p_j (1 - \Ebb_{\vc{V}}[R^\vc{V}_{n + 1, j}]).
    \]
    Moreover, we claim that
    \begin{claim}
        \label{cla:lbutil}
        \begin{equation}
            \Ebb_{\vc{V}} \left[\sum_{i = 1}^n u^\vc{V}(i)\right] \geq \sum_{j = 1}^m
            p_j \Ebb_{\vc{V}}[R^\vc{V}_{n + 1,j}].
        \end{equation}
    \end{claim}
    
    Note that \Cref{cla:lbutil} would be enough to prove the Lemma since
    \[
    \sum_{i = 1}^n \Ebb_{\vc{V}}[v_i(\vc{Z}^\vc{V}_i)] =
    \Ebb_{\vc{V}}[\rev^\vc{V}] +
    \Ebb_{\vc{V}}\left[\sum_{i = 1}^n u^\vc{V}(i)\right] \geq
    \sum_{j = 1}^m p_j(1 - \Ebb_{\vc{V}}[R^\vc{V}_{n + 1,j}] + \Ebb_{\vc{V}}[R^\vc{V}_{n + 1,j}])=
    \sum_{j = 1}^m p_j.
    \]
    And, by the definition of \(\vc{p}\) and \Cref{eq:sumSW}, we have
    \[
    \sum_{j = 1}^m p_j = \frac{1}{2} \Ebb_{\vc{V}}\left[\sum_{j = 1}^m SW^\vc{V}_j \right]
    = 
    \frac{1}{2} \sum_{i = 1}^n \Ebb_{\vc{V}}[V_i(\by^\vc{V}_i)]. \qedhere 
    \]
    \end{proof}
    Hence, we now focus on \Cref{cla:lbutil}.

    \begin{proof}[Proof of \Cref{cla:lbutil}]
        We start by proving an inequality that holds for all realizations of \(\vc{V}\).
        Let \(\vc{v} = (v_1, \dots, v_n) \in \Dcal\) be a valuation profile.
        Our plan to lower bound the utility of player \(i\) is to consider another fraction of items that was also available for \(i\). Since the player did not choose this vector, it means
        that the player's utility is at least that of this vector.

        Let \(\Tilde{\vc{v}} = (\Tilde{v}_1, \dots, \Tilde{v}_n) \in \Dcal\) such that \(\Tilde{v}_i = v_i\).
        Set \(\Tilde{y} \coloneqq y^{\Tilde{\vc{v}}}\).
        Define, for every \(i \in [n]\),
        \[ \textstyle
        \Tilde{q}_i = \int_0^1 \nabla v_i(t \Tilde{\by}_i) dt \in \Reals_+^m.
        \]
        Note that when player \(i\) arrived, at the execution with valuation profile \(\vc{v}\) and prices \(\vc{p}\),
        a feasible assignment would be
        \[
        \Tilde{Z}_{i, j} = R^\vc{v}_{i, j} \Tilde{y}_{i,j} \ones\{\Tilde{q}_{i, j} \geq p_j\} \qquad \text{for every \(j \in [m]\)}
        \]
        since \(\Tilde{y}_{i, j} \leq 1\) and so \(\Tilde{Z}_{i, j} \leq R^\vc{v}_{i, j}\).
        Using  \Cref{lema:DRprice} and that \(\Tilde{Z}_{i, j} \leq \Tilde{y}_{i, j}\), we obtain that the utility of the allocation \(\Tilde{Z}_i\) is
        \[
        v_i(\vc{\Tilde{Z}}_i) - \iprod{p}{\vc{\Tilde{Z}}_i} 
        \geq
    \textstyle        \sum_{j = 1}^m R^{\vc{v}}_{i, j} \Tilde{y}_{i, j} (\Tilde{q}_{i, j} - p_j)_+.
        \]
        Because player \(i\) chose \(Z^\vc{v}_i\) instead of \(\Tilde{Z}_i\), we have
        \[
        u^\vc{v}(i) \geq
        \textstyle \sum_{j = 1}^m R^{\vc{v}}_{i, j} \Tilde{y}_{i, j} (\Tilde{q}_{i, j} - p_j)_+.
        \]
        

        Observe that given a profile \(\vc{v}\), the last inequality holds for every choice of \(\Tilde{\vc{v}}\) such that \(\Tilde{v}_i = v_i\). Consider a random valuation profile
        \(\vc{\Tilde{V}}\) defined by sampling each valuation \(j \neq i\) independently from \(\Dcal_j\) and setting \(\Tilde{V_i} = v_i\). Taking expectation over \(\vc{\Tilde{V}}\) we obtain
        \[
        u^\vc{v}(i) \geq \sum_{j = 1}^m R^{\vc{v}}_{i, j} \, \textstyle \Ebb_{\vc{\Tilde{V}}}\left[\Tilde{y}_{i, j} (\Tilde{q}_{i, j} - p_j)_+\right].
        \]
        Note that we used that \(u^\vc{v}(i)\) and  \(R^{\vc{v}}_{i, j}\) does not depend on \(\vc{\Tilde{V}}\).
        
        Observe the previous inequality holds for every \(\vc{v} \in \Dcal\).
        Consider a random profile \(\vc{V}\) sampled from \(\Dcal\) and
        \(\vc{\Tilde{V}}\) where each \(\Tilde{V}_j\) for \(j \neq i\) is sampled independently from \(\Dcal_j\) and \(\Tilde{V}_i = V_i\).
        Observe that \(R^\vc{V}_{i, j}\) depends only on \(V_1, \dots, V_{i - 1}\), since it is the remaining before player \(i\)
        realization, and the terms \(\Tilde{y}\) and \(\Tilde{q}\) depends only on \(V_i\) and the resample \(\vc{\Tilde{V}}\),
        thus they are independent.
        We have
        \[
            \Ebb_{\vc{V}} \left[u^\vc{V}(i)\right] \geq
            \sum_{j = 1}^m \Ebb_{\vc{V}} \left[R^{\vc{V}}_{i, j}\right]
\textstyle            \, \Ebb_{(V_i, \vc{\Tilde{V}})}\left[\Tilde{y}_{i, j} (\Tilde{q}_{i, j} - p_j)_+\right].
        \]
        Using the fact that \(\vc{R}^\vc{V}\) is decreasing in \(i\) and that \(\alpha_+ \geq \alpha\) for any \(\alpha \in \Reals\), we get
        \[
            \Ebb_{\vc{V}} \left[u^\vc{V}(i)\right] \geq
          \textstyle  \sum_{j = 1}^m \Ebb_{\vc{V}} \left[R^{\vc{V}}_{n+1, j}\right]
            \, \Ebb_{(V_i, \vc{\Tilde{V}})}\left[\Tilde{y}_{i, j} (\Tilde{q}_{i, j} - p_j)\right].
        \]
        However, we have that sampling \((V_i, \vc{\Tilde{V}})\), i.e., sampling \(V_i\) and then \(\vc{\Tilde{V}}\) is equivalent to sample \(\vc{V} \sim \Dcal\), thus
        \[
            \Ebb_{(V_i, \vc{\Tilde{V}})}\left[\Tilde{y}_{i, j} (\Tilde{q}_{i, j} - p_j)\right] =
            \Ebb_{\vc{V}} \left[y^\vc{V}_{i, j} (q^\vc{V}_{i, j} - p_j)\right].
        \]

        Summing the bound for every player, we get
        \begin{align*}    \textstyle
            \Ebb_{\vc{V}}\left[ \sum_{i = 1}^n u^\vc{V}(i)\right] &\geq \textstyle
            \sum_{i = 1}^n \sum_{j = 1}^m
            \Ebb_{\vc{V}} \left[R^{\vc{V}}_{n+1, j}\right] \Ebb_{\vc{V}} \left[y^\vc{V}_{i, j} (q^\vc{V}_{i, j} - p_j)\right]\\
            &= \textstyle
            \sum_{j = 1}^m \Ebb_{\vc{V}} \left[R^{\vc{V}}_{n+1, j}\right] \Ebb_{\vc{V}} \left[ \sum_{i = 1}^n y^\vc{V}_{i, j} (q^\vc{V}_{i, j} - p_j)\right].
        \end{align*}
        Looking more carefully at the second expectation, it holds that
        \[ \textstyle
        \Ebb_{\vc{V}} \left[ \sum_{i = 1}^n y^\vc{V}_{i, j} (q^\vc{V}_{i, j} - p_j)\right]
        =
        \Ebb_{\vc{V}} \left[ \sum_{i = 1}^n  y^\vc{V}_{i, j} q^\vc{V}_{i, j} \right] -
        p_j \Ebb_{\vc{V}}\left[\sum_{i = 1}^n y^\vc{V}_{i, j}\right]
        \geq
        \Ebb_{\vc{V}} [SW^\vc{V}_j] - p_j,
        \]
        where in the last inequality we used \Cref{eq:defSW} and the hypothesis that
        \(\Ebb_{\vc{V}}\left[\sum_{i = 1}^n y^{\vc{V}}_{i, j}\right] \leq 1\).
        Plugging in the definition of \(p_j\), we obtain
        \[ \textstyle
        \Ebb_{\vc{V}}\left[ \sum_{i = 1}^n u^\vc{V}(i)\right] \geq
        \sum_{j = 1}^m \Ebb_{\vc{V}} \left[R^{\vc{V}}_{n+1, j}\right] (2p_j - p_j)
        =
        \sum_{j = 1}^m p_j \Ebb_{\vc{V}} \left[R^{\vc{V}}_{n+1, j}\right].
        \]
        Which completes the proof of \Cref{cla:lbutil}.
    \end{proof}

Finally, we can use \Cref{lem:generalphi} to prove that, if computational efficiency is not a requirement, there exists a pricing vector \(\vc{p}\)
achieving a \(1/2\)-approximation even if we remove the concavity requirement from the valuation functions.

\begin{lemma}
    For the online welfare with divisible items (\OWDI) problem and monotone DR-submodular valuations, there exist linear per-unit item prices such that the fixed-price auction achieves a \(1/2\)-approximation to the offline/prophet benchmark.
\end{lemma}
\begin{proof}
    For every \(i \in [n], j \in [m]\) and \(\vc{v} \in \Dcal\) define \(y^\vc{v}_{i, j}\) as the optimal allocation of item \(j\) to player \(i\) when \(\vc{v}\) is sampled from \(\Dcal\). Since each item has a unit supply, we have \(\sum_{i = 1}^n y^\vc{v}_{i, j} \leq 1\) for every item and every valuation. Thus, \(y\) satisfies the hypothesis of \Cref{lem:generalphi}.
    Moreover, \(\Ebb_{\vc{V}} \left[ \sum_{i = 1}^n  {V}_i (\by^\vc{V}_i)\right]\) is exactly the expected optimum. Hence, by \Cref{lem:generalphi}, there are prices such that \Cref{algo:postedprice} is \((1/2)\)-competitive against the offline/prophet benchmark.
\end{proof}

\section{Hardness}
\label{sec:hardness}

In this section, we show the  {\textbf{\#P}}-hardness of computing the optimal policy against the philosopher/online benchmark and the information-theoretic hardness of factor $\frac12$ against the prophet/offline benchmark.

\subsection{Philosopher Hardness}
We will prove that finding the optimal online policy for \OWDI for monotone concave DR-submodular valuations is \sharpp, even in the special case of a single item. Our proof will reduce from \textbf{\#}Knapsack, a known \sharpp problem.

\begin{definition}[\textbf{\#}Knapsack]
    An instance of \textbf{\#}Knapsack consists of  \(n+1\) nonnegative integers \(L\), \(a_1, \dots, a_n\).
    The objective is to count the number of sets \(S \subseteq [n]\) such that
    \[\sum_{i \in S} a_i \leq L\]
\end{definition}

\begin{proposition}[\cite{JerrumSinclair1996MCMC,valiant1979complexity}]
    \textbf{\#}Knapsack is \sharpp.
\end{proposition}

\begin{theorem}
    Finding the optimal online algorithm for \OWDI with monotone concave DR-submodular functions is \sharpp, even for the special case of a single item.
\end{theorem}
\begin{proof}
    We will show that if given access to an optimal online policy for
    \OWDI  that runs in polynomial time, we can count the number of knapsack solutions in polynomial time. More formally, the access to the optimal policy is given by a polynomial time oracle such that given an \OWDI instance and a list of past realizations \(V_1, \dots, V_i\) decide how much to assign to player \(i\). In fact, for our instance it will be enough to make this decision for the first player and there is only one function in his distribution.

    Let \(L, a_1, \dots, a_n\) be the input for a \textbf{\#}Knapsack instance. 
    Set \(L' \coloneqq L + 1\).
    For a set \(I \subseteq [n]\), we define \(a(I) \coloneqq \sum_{i \in I} a_i\).
    Consider the \(\OWDI\) instance with \(n + 1\) players numbered from \(0\) to \(n\) and \(1\) item.
    Player \(0\) samples with probability \(1\) the function \(v_0(x) = wx\), where \(w \in \Reals_+\) will be set later.
    Player \(i \in [n]\) samples with probability \(1/2\) the zero function and otherwise samples  \( v_i(x) = L' \min\{a_i/L', x\}\).
    When  player  \(i \in [n]\) samples the nonzero function, we say player $i$ is \emph{active}.
    Note that all functions described are monotone, concave, and DR-submodular.

    The proof is based on the idea that computing the expected value of the optimal online policy for players $1$ to $n$ is \sharpp. Hence, if for any choice of parameter $w$ we can decide whether to allocate a non-zero fraction of the item to player $0$, then we can perform a binary search over $w$ to solve a \sharpp problem.

    We begin by calculating the expected value of an optimal policy that leaves $R\in [0,1]$ fraction of the item to players $1$ to $n$. 

    \begin{claim}
    \label{cla:optpol}
        Suppose that after the first player there is still \(R \in [0, 1]\) amount of the item. Then, the expected value obtained by the optimal online policy for players \(1\) to \(n\) is
        \[\Ebb[\min\{RL', a(S)\}],\]
        where \(S\) is a set drawn uniformly at random from \([n]\).
    \end{claim}
    \begin{proof}
        When a player is active, the maximum the objective value can increase if we assign a \(\rho\) fraction of the item to him is \(\rho L'\), i.e., all players increase the objective by the same ``rate''. Therefore, whenever a player is active we should assign the item to him until it reaches his cap at \(a_i/L'\)
        or the item is over, since no player in the future can contribute more.
        Thus, if the set of active players is \(S\), the objective value of the optimum policy is
        \[
            L' \min \{R, a(S)/L'\} = \min \{RL', a(S)\}
        \]
        where the first term achieves the minimum when we are able to assign the whole remaining of the item, and the second term achieves the minimum
        if the sum of the active players capacities is less than the remaining of the item.
        A player is active with probability \(1/2\) and independent of the others, and so the set of active players is uniformly distributed among the subsets of \([n]\), and this concludes the proof of the claim. 
    \end{proof}

    As a consequence of \Cref{cla:optpol}, the expected value of the
    optimal online policy is
    \[\max_{R \in [0,1]}\{w(1 - R) + \Ebb[\min\{RL', a(S)\}]\}.\]
    For each \(w \in \Reals_+\), let \(R(w)\) be the value of \(R\)
    that maximizes the expected value of the optimal online policy.
    The following claim characterizes when the optimal online policy allocates a non-zero fraction of the item to player $0$.

    \begin{claim}
    \label{cla:iffRw}
    If \(w < L' \Pr[a(S) \geq L']\), then any optimal online policy allocates zero amount of the item to player \(0\) (i.e., $R(w) = 1$).
    If \(w > L' \Pr[a(S) \geq L']\), then any optimal online policy allocates some amount of the item to player \(0\) (i.e., $R(w) < 1$).
    \end{claim}
    \begin{proof}
        Suppose $w < L' \Pr[a(S) \geq L']$. We will prove that \(R(w) = 1\)  by showing that for any \(R < 1\), the expected value of any online policy that allocates $1 - R$ to player $0$ satisfies
         \[w(1 - R) + \Ebb[\min\{RL', a(S)\}] < \Ebb[\min\{L', a(S)\}].\]
        Expanding the expectations on both sides and using the bound on \(w\), it suffices to show
        \begin{align*}
            L' \Pr[a(S) \geq L'](1 - R) + RL' \Pr[a(S) \geq RL'] ~+~ &   \frac{1}{2^n}\sum_{S: a(S) < RL'} a(S) \\\leq \;
            &L' \Pr[a(S) \geq L']
             + \frac{1}{2^n}\sum_{S: a(S) < L'} a(S)    \enspace .
        \end{align*}
        Grouping similar terms, we obtain it is equivalent to showing
        \[
        RL' \Pr[L' > a(S) \geq RL']~ \leq~  \frac1{2^n} \sum_{S: RL' \leq a(S) < L'} a(S) \enspace.
        \]
        This inequality is true since each set \(S\) satisfying
        \(RL' \leq a(S) < L'\) contributes  \(RL'/2^n\) on the LHS and contributes \(a(S)/2^n\) on the RHS.

        Now suppose that \(w >  L' \Pr[a(S) \geq L']\).
        Choose \(\varepsilon > 0\) sufficiently small\footnote{Note that \(\varepsilon = 1/(2L')\) already suffices since there is no integer \(k\) in \([L' - 1/2, L')\) for integer \(L'\).}
        so that
        there is no set \(S\) satisfying \(L' > a(S) \geq (1 - \varepsilon)L'\). 
        We claim that setting \(R = 1-\varepsilon\) gives a bigger value than setting \(R = 1\), hence \(R(w) \neq 1\).
        It is sufficient to prove that
        \begin{align*}
            \varepsilon w + (1 - \varepsilon)L' \Pr[a(S) \geq (1 - \varepsilon)L'] ~+~ &\frac{1}{2^n}\sum_{S: a(S) < (1 - \varepsilon)L'} a(S) \\> \;
            &L' \Pr[a(S) \geq L']
            + \frac{1}{2^n}\sum_{S: a(S) < L'} a(S).
        \end{align*}
        By our condition on \(\varepsilon\), if \(a(S) < L'\) then
        \(a(S) < (1 - \varepsilon)L'\), and similarly, if \(a(S) \geq (1 - \varepsilon)L'\) then \(a(S) \geq L'\).
        By canceling similar terms, the inequality simplifies to
        \[\varepsilon w > \varepsilon L' \Pr[a(S) \geq L'] \enspace, \]
        which is precisely our hypothesis.
    \end{proof}

    If we can obtain the optimal online policy for the described instance of \OWDI, then with \Cref{cla:iffRw} we can use binary search on \(w\) to estimate \(L' \cdot \Pr[a(S) \geq L']\). Since
    \(\Pr[a(S) \geq L']\) is a multiple of \(1/2^n\), with polynomial in \(n\) and \(\log L'\) iterations of binary search we can determine exactly the number of sets \(S\) such that \(a(S) \geq L'\).
    Finally, observe that the number of sets with \(a(S) \geq L'\) is the complement of the sets with \(a(S) \leq L\), which corresponds to the number of Knapsack solutions.
\end{proof}

\subsection{Prophet Hardness}

\begin{theorem}
    No online algorithm can achieve a competitive ratio strictly better than \(1/2\) for \OWDI even for the special case of $m=1$ item and $n=2$ players with linear valuations.
\end{theorem}
\begin{proof}
    The proof is similar to the standard factor $1/2$ hardness for prophet inequality with a single indivisible item. 
    We will prove that for every \(\varepsilon > 0\) achieving \(\tfrac{1}{2} + \varepsilon\) is impossible.
    Consider an instance with two players and one item.
    For the first player, the only function in the support of \(\Dcal_1\) is \(v_1(x) = x\).
    For the second player, they sample the function \(v_2(x) = x/\varepsilon\) with probability \(\varepsilon\) and the function
    identically equals \(0\) with probability \(1 - \varepsilon\).

    For the offline optimum, we get \(\varepsilon^{-1}\) if the second player samples \(v_2\), or else we give the whole item to the first player if the second player gets the \(0\) function. Hence,
    \[\Ebb[\OPT] = \varepsilon \varepsilon^{-1} + (1 - \varepsilon) = 2-\varepsilon.\]

    Consider now an online algorithm \(\ALG\).
    Since the input is coming from a fixed distribution, we may assume that the optimal online algorithm is deterministic.
    Denote by \(\alpha\) how much \(\ALG\) assigns to player \(1\).
    Hence,
    \[\Ebb[\ALG] \leq \alpha + \varepsilon \frac{(1 - \alpha)}{\varepsilon} = 1.\]
    Thus, 
        $\frac{\Ebb[\ALG]}{\Ebb[\OPT]} \leq \frac{1}{2 - \varepsilon}$.
    \end{proof}

\paragraph{Use of AI.}
The authors used ChatGPT 5.4 and 5.5 for assistance with writing, developing parts of the proof of \Cref{lema:decphi}, and in search of scaling functions as presented in \Cref{sec:appendixgoodf}. The authors have checked the arguments and take full responsibility for the contents of the paper.

\clearpage
\appendix

\section{DR-submodularity}
\label{sec:appendixDR}
\DRbasic*
\begin{proof}
    We shall prove that definitions \((1)\) and \((3)\) are equivalent to \((2)\).

    We start with \((1) \Leftrightarrow (2)\). Let \(\bx, \by \in \Reals_+^m\) such that \(\by \leq \bx\) and let \(i \in [m]\).
    To prove the forward direction, note that
    \[
        \nabla f(\by)_i = \lim_{h \to 0} \frac{f(\by + he_i) - f(\by)}{h}
        \geq
        \lim_{h \to 0} \frac{f(\bx + he_i) - f(\bx)}{h} = \nabla f(\bx)_i,
    \]
    where the inequality comes from definition \((1)\).
    To prove the reverse direction, note that it is sufficient to prove for vectors \(\bz = \gamma \vc{e_i}\) with \(\gamma \in \Reals_+\) and then apply the result inductively.
    It holds that
    \[
    f(\by + \gamma \vc{e_i}) - f(\by) = \int_0^{\gamma} \nabla f(\by + t\vc{e_i})^\top \vc{e_i} \; dt
    \geq \int_0^{\gamma} \nabla f(\bx + t\vc{e_i})^\top \vc{e_i} \; dt = f(\bx + \gamma \vc{e_i}) - f(\bx).
    \]

    We now prove \((2) \Leftrightarrow (3)\).
    Let \(\bx \in \Reals^m_+\) and let \(i, j \in [m]\).
    For the forward direction, it holds that
    \[
    \frac{\partial^2 f(\bx)}{\partial x_i \partial x_j} =
    \frac{\partial}{\partial x_i} (\nabla f(\bx))_j
    =
    \lim_{h \to 0_+} \frac{(\nabla f(\bx + h \vc{e_i}))_j - (\nabla f(\bx))_j}{h} \leq 0.
    \]
    The last inequality comes from the fact that \(\bx \leq \bx + h\vc{e_i}\) implies \(\nabla f(\bx + h\vc{e_i}) \leq \nabla f(\bx)\).
    For the reverse direction, let \(\gamma \in \Reals_+\) and note that it is sufficient to prove that \(\nabla f(\bx)_j \geq \nabla f(\bx + \gamma \vc{e_i})_j\).
    Define \(g(t) = \nabla f(\bx + t\vc{e_i})_j\). Then,
    \[
    \nabla f(\bx + \gamma \vc{e_i})_j - \nabla f(\bx)_j = g(\gamma) - g(0)
    = \int_0^\gamma \nabla^2f(\bx + u\vc{e_i})_{ij} \; du \leq 0.
    \]
    Rearranging the terms give us the desired inequality.
\end{proof}

\section{Benchmark Relaxations}
\label{sec:appendixLP}

For convenience of the reader, we repeat the configuration LP from \Cref{sec:prelim}.

\configrelax*


\begin{lemma} \label{lem:configLPrelax}
    The \ref{eq:confLP} is a relaxation of the online benchmark.
\end{lemma}
\begin{proof}
    Let \(\OPT\) be an optimal online policy.
    For every \(i \in [n]\) and \(\bz \in [0,1]^m\), let \(Y_{i, \bz}\) be the event that player \(i\)
    receives bundle \(\bz\) from \(\OPT\).
    We claim that
    \[
    y_{i, k, \bz} \coloneqq \Pr[Y_{i, \bz} \cap V_i = v_{i, k}] \qquad \forall i \in [n], k \in [d_i], \bz \in [0,1]^m
    \]
    is feasible and has objective value equal to
    \(\Ebb[\OPT]\).

    Let \(i \in [n]\) and \(k \in [d_i]\).
    It holds that,
    \[\textstyle
    \sum_{\bz \in [0,1]^m} \Pr[Y_{i, \bz} \cap V_i = v_{i, k}] ~=~ \Pr[V_i = v_{i, k}] ~=~ p_{i,k},
    \]
    hence the first constraint is satisfied.
    For the online constraint, let \(j \in [m]\) and
    note that
    for any realization of the valuations (and possibly of any randomness of \(\OPT\)) it must hold that
    \[\textstyle
    \sum_{\bz \in [0,1]^m} z_j \1_{Y_{i, \bz}}
    ~\leq~ 
    1 - \sum_{\ell < i} \sum_{\bz \in [0,1]^m} z_j \1_{Y_{\ell, \bz}}
    ~=~
    1 - \sum_{\ell < i} \sum_{r = 1}^{d_\ell} \sum_{\bz \in [0,1]^m} z_j \1_{Y_{\ell, \bz}}\1_{V_\ell = v_{\ell,r}},
    \]
    where in the inequality we used that each item has a unit supply and in the equality we used that only one function is sampled by player \(\ell\).
    Consequently, it holds that
    \[\textstyle
    \sum_{\bz \in [0,1]^m} z_j \1_{Y_{i, \bz}} \1_{V_i = v_{i, k}}
    \leq
    \1_{V_i = v_{i, k}}
    \Big(
    1 - \sum_{\ell < i} \sum_{r = 1}^{d_\ell} \sum_{\bz \in [0,1]^m} z_j \1_{Y_{\ell, \bz}}\1_{V_\ell = v_{\ell,r}}
    \Big).
    \]
    By taking expectation on both sides and using that \(Y_{\ell, \bz}\) and \(\ones_{V_\ell = v_{\ell, r}}\) are independent of \(\ones_{V_i = v_{i, k}}\), we obtain
    \[\textstyle
    \sum_{\bz \in [0,1]^m} z_j y_{i, k, \bz}
    \leq
    p_{i,k}
    \Big(
    1 - \sum_{\ell < i} \sum_{r = 1}^{d_\ell} \sum_{\bz \in [0,1]^m} z_j y_{\ell, r, \bz}
    \Big).
    \]
    Thus, the \(y\) constructed is feasible.

    Finally, we have that for any realization
    \[\textstyle
    \OPT = \sum_{i = 1}^n \sum_{k = 1}^{d_i} \sum_{\bz \in [0,1]^m} \ones_{Y_{i, \bz}} \ones_{V_i = v_{i, k}} v_{i, k}(\bz),
    \]
    and by taking expectation on both sides, it holds that
    \[\textstyle
    \Ebb[\OPT] = \sum_{i = 1}^n \sum_{k = 1}^{d_i} \sum_{\bz \in [0,1]^m} y_{i, k, \bz} v_{i, k}(\bz),
    \]
    which is the objective value of \(y\).
\end{proof}

\offlinerelax*
\begin{proof}
It is enough to show that for each feasible solution of a program one can construct a feasible solution to the other one with better objective value.

Let \(y_{i, k, \bz}\) be a feasible solution for \eqref{eq:confLP}. Set
\[ \textstyle
\hat x_{i, k, j} \coloneqq \sum_{\bz \in [0,1]^m} z_j y_{i, k, \bz}
\qquad \forall i \in [n], k \in [d_i], j \in [m].
\]
We first verify that the constructed \(\hat x\) is feasible.
For every \(i \in [n]\), \(k \in [d_i]\), and \(j \in [m]\), 
\[\textstyle
\hat x_{i, k, j} = \sum_{\bz \in [0,1]^m} z_j y_{i, k, \bz}
\leq
p_{i, k} \left(1 - \sum_{\ell < i} \sum_{r = 1}^{d_\ell} \sum_{\bz \in [0,1]^m} z_j y_{\ell, r, \bz}\right)
=
p_{i, k} \left(1 - \sum_{\ell < i} \sum_{r = 1}^{d_\ell} x_{\ell, r, j}\right) .
\]
Regarding the objective value, for every \(i \in [n]\) and \(k \in [d_i]\), we have
\begin{align*}\textstyle
    \sum_{\bz \in [0,1]^m } y_{i, k, \bz} v_{i, k}(\bz)
    &= \textstyle
    p_{i, k}
    \sum_{\bz \in [0,1]^m } \frac{y_{i, k, \bz}}{p_{i, k}} v_{i, k}(\bz)\\
    &\textstyle \leq
    p_{i, k}
    v_{i, k}
    \Big( \sum_{\bz \in [0,1]^m} \frac{y_{i, k, \bz}}{p_{i, k}} \bz
    \Big)&&\text{by Jensen's inequality}\\
    & \textstyle =
    p_{i,k} v_{i, k} (\vc{\hat x_{i, k}}/p_{i, k}).
\end{align*}
For the Jensen's inequality step, we can assume that \(\sum_{\bz}y_{i,k,\bz} = p_{i,k}\), by increasing the value assigned to empty bundle \(\bz = \vc{0}\).

For the other direction, let \(x_{i, k, j}\) be a feasible solution for \eqref{eq:onLPm}.
Set
\[
\hat y_{i, k, \bz } \coloneqq p_{i,k} \ones[\bz = \vc{x_{i,k}}/p_{i,k}]
\qquad \forall i \in [n], k \in [d_i], \bz \in [0,1]^m.
\]
Note that \(\vc{x_{i,k}}/p_{i, k} \in [0,1]^m\).
The first constraint holds trivially, since there is only one vector \(\bz\) in the support of \(\vc{\hat y_{i, k}}\).
For the online constraint,
\[\textstyle
\sum_{\bz \in [0,1]^m} z_j \hat y_{i, k, \bz} =
\frac{x_{i, k, j}}{p_{i, k}} p_{i, k} = x_{i, k, j},
\]
hence, it is also satisfied.
Finally, it follows from the definition that
\[\textstyle
\sum_{\bz \in [0,1]^m} \hat y_{i, k, \bz} v_{i, k}(\bz)
=
p_{i, k} v_{i, k}(\vc{x_{i, k}}/p_{i,k})
\]
for every \(i \in [n]\) and \(k \in [d_i]\).
\end{proof}

\section{\texorpdfstring{\good Functions}{}}
\label{sec:appendixgoodf}

\subsection{\texorpdfstring{$\f(z) = 1 - \tfrac{1}{z+1}$}{z/(z+1)}}

Observe that \(\f(1) = \nicefrac{1}{2}\), we now prove it is a scaling function.

It is easy to see that \(\f(z) \in \Reals_+\) for every \(z \in \Reals_+\).
The function \(\f\) is concave since \(\f''(z) = \frac{-2}{(z+1)^3} < 0\).
Furthermore, \(\f(z)\) is clearly less than \(1\) and
\(\f(z)\) can be written as \(\tfrac{z}{z + 1} \leq z\).

It is left to prove the inequality
\[
(1 - \mu) \left(1 - \frac{1}{\frac{z-\f(z)}{\mu} + 1}\right) + \mu\left(1 - \frac{1}{\frac{z}{\mu} + 1}\right)
\geq
1 - \frac1{z+1},
\]
for every \(z \in \Reals_+\) and \(\mu \in (0,1]\).
This is equivalent to
\[
(1 - \mu) \left(1 - \frac{\mu}{{z-\f(z)} + \mu}\right) + \mu\left(1 - \frac{\mu}{z + \mu}\right)
\geq
1 - \frac1{z+1},
\]
which can be simplified to
\[
\frac{(1-\mu)(z-\f(z))}{z - \f(z) + \mu}
+
\frac{ z\mu }{z + \mu}
\geq
\frac{z}{z+1}.
\]

Because \(\f(z) \geq 0\), it is sufficient to prove
\[
\frac{(1-\mu)(z-\f(z))}{z + \mu}
+
\frac{ z\mu }{z + \mu}
\geq
\frac{z}{z+1}.
\]
Multiplying both sides by \((z +\mu)(z+1)\), we get
\[
(1 - \mu)\left(z - 1 + \frac{1}{z+1}\right)(z+1)
+
z\mu(z+1) \geq z(z + \mu).
\]
We claim that the last inequality holds with equality.
Staring from the LHS, we have
\begin{align*}
(1 - \mu)\left(z - 1 + \frac{1}{z+1}\right)(z+1)
+
z\mu(z+1)
&=
(1 - \mu) ((z-1)(z+1) + 1) + z^2\mu + z\mu\\
&=
(1 - \mu) z^2 + z^2\mu +z\mu\\
&=
z^2 + z\mu
=z(z + \mu).
\end{align*}

\subsection{\texorpdfstring{$\f(z) = 1-e^{-z}$}{1-e-z}}

Observe that \(\f(1) = 1 - \nicefrac{1}{e}\), we now prove it is a scaling function.

The function only maps to nonnegative values since \(e^{-z} \leq 1\) for every \(z \in \Reals_+\).
The function is concave since \(\f''(z)= - e^{-z} \leq 0\).
Moreover, it clearly is less than or equal to \(1\), and it is less than \(z\) by using the standard inequality \(1 + z \leq e^{z}\) for every \(z \in \Reals\).

We now focus on proving the second constraint.
     Let \(z \in \Reals_+\) and \(\mu \in (0,1]\).
     We wish to prove that
     \[
     (1 - \mu)
     \left(1 - \exp\left(\frac{1 - e^{-z} - z}{\mu}\right)\right) + \mu(1 - e^{-z/\mu})
     \geq
     1-e^{-z},
     \]
     which can be simplified to
     \[
     (1 - \mu)
     \exp\left(\frac{1 - e^{-z} - z}{\mu}\right) + \mu e^{-z/\mu}
     \leq
     e^{-z}.
     \]
     Multiplying both sides by \(e^{z/\mu}\), we get
     \[
     (1 - \mu)
     \exp\left(\frac{1 - e^{-z}}{\mu}\right) + \mu 
     \leq
     \exp\left(z\left(\frac1{\mu} - 1\right)\right).
     \]
     At \(z = 0\), both sides equal \(1\).
     We shall prove that the RHS has a bigger derivative with respect to \(z\) for all \(z \geq 0\).
     Taking the derivative on both sides, our target inequality becomes
     \[
     \left(\frac1{\mu} - 1\right)
     \exp\left(\frac{1 - e^{-z}}{\mu} - z\right)
     \leq
     \left(\frac1{\mu} - 1\right)\exp\left(z\left(\frac1{\mu} - 1\right)\right),
     \]
     which is true as long as
     \[
     \frac{1 - e^{-z}}{\mu} - z
     \leq
     \frac{z}{\mu} - z,
     \]
     which holds since \(1 - e^{-z} \leq z\).

\subsection{\texorpdfstring{$\f(z) = 1 - \tfrac1{1+z+z^2/2 + z^3/2}$}{1-1/(1+z+z2/2 + z3/2)}}

Observe that \(\f(1) = 1 - \tfrac{1}{3} = \tfrac{2}{3}\). Our goal now is to prove that \(\f\) is a \good function.

Define \(p(z) = 1 + z + \tfrac{z^2}{2} + \tfrac{z^3}{2}\) and \(\f(z) = 1 - \tfrac1{p(z)}\).
Since \(p(z) \geq 1\) for all \(z \in \Reals_+\),
we have \(\f(z) \in \Reals_+\) for all \(z \in \Reals_+\).
The first step is to prove it is concave.
We have that
\[
\f''(z) = \frac{p''(z)p(z) - 2(p'(z))^2}{p(z)^3}.
\]
Since \(p(z)^3 \geq 0\), to show that \(\f\) is concave, it suffices to prove that
\[
2(p'(z))^2 \geq p''(z)p(z) \qquad \forall z \in \Reals_+.
\]
This is equivalent to
\[
2(1 + z + \tfrac{3}{2}z^2)^2 \geq (1 + 3z)(1 + z + \tfrac{z^2}{2} + \tfrac{z^3}{2}),
\]
which simplifies to checking 
\[
6z^4 + 8z^3 + 9z^2 + 2 \geq 0
\]
that is clearly true for $z\geq 0$. 

We now prove that \(\f(z) \leq \min\{z, 1\}\) for all \(z \in \Reals_+\), it is clearly less than \(1\) since \(p(z) > 0\).
It is less than \(z\) if and only if
\[
p(z) - 1 \leq z p(z)
\]
which is equivalent to
\[
z + \frac{z^2}{2} + \frac{z^3}{2} \leq z + z^2 + \frac{z^3}{2} + \frac{z^4}{2}.
\]
The last inequality holds since after canceling terms we are left with \(z^2 + z^4 \geq 0\).

The last step is to show that for every \(z \in \Reals_+\) and \(\mu \in (0,1]\), it satisfies
\[
(1 - \mu) \f \left(\frac{z - \f(z)}{\mu}\right)
+
\mu \f\left(\frac{z}{\mu}\right) \geq \f(z).
\]
By expanding \(\f\) and canceling terms, we obtain it is sufficient to prove
\[
\frac{1-\mu}{p\left(\frac{z-\f(z)}{\mu}\right)}
+
\frac{\mu}{p \left(\frac{z}{\mu}\right)}
\leq
\frac1{p(z)}.
\]
Define \(q(z) \coloneqq \tfrac12 z^2 (z^2 + 1)\)
and note that
\[
z - \f(z)
=
z - 1 + \frac{1}{1 + z + z^2/2 + z^3/2}
=
\frac{
(z-1)(1 + z + z^2/2 + z^3/2) + 1
}{1 + z + z^2/2 + z^3/2}
=
\frac{q(z)}{p(z)}.
\]
Define \(\alpha \coloneqq 1/\mu\), it is sufficient to prove
\begin{equation}
\label{eq:C1}
F_1(z, \alpha) \coloneqq 
\frac{\alpha}{p(z)}-
\frac{1}{p \left(\alpha z\right)}
-
\frac{\alpha - 1}{p\left(\frac{\alpha q(z)}{p(z)}\right)}
\geq 
0
\qquad \forall z \geq 0 \text{ and }\forall \alpha \geq 1.
\end{equation}

Define \(r(z) \coloneqq p^3(z) + \alpha q(z)p^2(z) + \frac{1}{2} \alpha^2 q^2(z) p(z) + \frac12 \alpha^3 q(z)^3 \) and note that
\[
r(z) = p^3(z) \left(1 + \alpha \frac{q(z)}{p(z)} +
\frac{\alpha^2 q^2(z)}{2 p(z)^2} + \frac{\alpha^3 q^3(z)}{2 p^3(z)}\right)
=
p^3(z)p\left(\frac{\alpha q(z)}{p(z)}\right).
\]
Hence,
\[
F_1(z, \alpha)
=
\frac{\alpha}{p(z)}-
\frac{1}{p \left(\alpha z\right)}
-
\frac{(\alpha - 1)p^3(z)}{r(z)}
=
\frac{
\alpha p(\alpha z) r(z)
- p(z) r(z)
- (\alpha - 1)p^4(z)p(\alpha z)
}{p(z)p(\alpha z) r(z)}.
\]
Since \(p(z) \geq 1\) and \(r(z) \geq 1\) for all \(z \in \Reals_+\) we have that the denominator is always positive, thus it suffices to analyze the numerator.
We start by simplifying the expression:
\begin{align*}
\alpha p(\alpha z) - p(z)
&=
(\alpha - 1) + (\alpha^2 - 1)z
+ (\alpha^3 - 1) \frac{z^2}{2}
+ (\alpha^4 - 1) \frac{z^3}{2}\\
&=
(\alpha - 1)
\left(1 + (\alpha + 1)z
+ (\alpha^2 + \alpha + 1) \frac{z^2}{2}
+ (\alpha^3 + \alpha^2 + \alpha + 1) \frac{z^3}{2}\right).
\end{align*}
Thus, by factoring out the term \(\alpha - 1\), which is always non-negative since \(\alpha \geq 1\), we obtain that it is sufficient to prove
\[
F_2(z, \alpha) \coloneqq
r(z)\left(1 + (\alpha + 1)z
+ (\alpha^2 + \alpha + 1) \frac{z^2}{2}
+ (\alpha^3 + \alpha^2 + \alpha + 1) \frac{z^3}{2}\right)
-p^4(z) p(\alpha z) \geq 0.
\]
It holds that
\[
\begin{aligned}
F_2(z, \alpha)
&= \frac{z^3 \alpha^2}{32}
\Big(
z^{12}\alpha^{4}
+ z^{12}\alpha^{3}
+ z^{12}\alpha^{2}
+ 2z^{11}\alpha^{3}
+ 3z^{10}\alpha^{4}
+ 2z^{11}\alpha^{2} \\
&+ 4z^{10}\alpha^{3}
- 2z^{11}\alpha
+ 9z^{10}\alpha^{2}
+ 7z^{9}\alpha^{3}
+ 3z^{8}\alpha^{4}
- 5z^{10}\alpha \\
&+ 12z^{9}\alpha^{2}
+ 7z^{8}\alpha^{3}
- 16z^{9}\alpha
+ 29z^{8}\alpha^{2}
+ 8z^{7}\alpha^{3}
+ z^{6}\alpha^{4} \\
&+ z^{9}
- 32z^{8}\alpha
+ 32z^{7}\alpha^{2}
+ 6z^{6}\alpha^{3}
+ 4z^{8}
- 52z^{7}\alpha \\
&+ 43z^{6}\alpha^{2}
+ 3z^{5}\alpha^{3}
+ 10z^{7}
- 75z^{6}\alpha
+ 40z^{5}\alpha^{2}
+ 2z^{4}\alpha^{3} \\
&+ 24z^{6}
- 82z^{5}\alpha
+ 30z^{4}\alpha^{2}
+ 37z^{5}
- 76z^{4}\alpha
+ 18z^{3}\alpha^{2} \\
&+ 52z^{4}
- 56z^{3}\alpha
+ 8z^{2}\alpha^{2}
+ 60z^{3}
- 28z^{2}\alpha
+ 48z^{2} \\
&- 8z\alpha
+ 36z
+ 16
\Big).
\end{aligned}
\]
See this \href{https://sagecell.sagemath.org/?z=eJxlkE1OAzEMhfeRcger3SQuMGLSFWrZAHs6BxgUqSkaqc14fpAgp8dOQylikcT6np-flSU8h0MXA1B__Ir9qfNHGLv4rlVzt0k3_hG28HqRGlbMbme10mr53zlpRdx_DytIctq6qvPrqlqrgaVftBY0MqLWMfA4ILUi-dyCMLRyUyauEKfVXAKMXN5iutRnL5OScU3PU0TLq8j2T_2JPuYA4ZPGME1dH7WSmue7Gs2MI9zycmskk7Yek7WVYTtKhkzgXh_3Yf9WTI2RIkvEPzibxUvpuIp4WNgf-Y_ffgOam2ud&lang=sage&interacts=eJyLjgUAARUAuQ==}{link} for a computer verification of the last equality.
Since \(z, \alpha \geq 0\) it is enough to consider the expression inside the parenthesis,
we call it \(F_3(z,\alpha).\)

Define
\[
\begin{aligned}
F_4(z,y)
\coloneqq{}&
z^4\!\left(3y^4+8y^3+67y^2+7y\right)
+z^5\!\left(19y^3+56y^2+33y\right) \\
&+z^6\!\left(3y^4+7y^3+68y^2+36y\right)
+z^7\!\left(16y^3+33y^2+59y\right) \\
&+z^8\!\left(y^4+2y^3+39y^2+29y\right)
+z^9\!\left(5y^3+8y^2+37y+4\right) \\
&+z^{10}\!\left(10y^2+8y+11\right)
+z^{11}(9y+2)
+3z^{12}.
\end{aligned}
\]
and
\[
\begin{aligned}
F_5(z,y)
\coloneqq{}&
8(y^2-y+2)
+2z\!\left(y^3+9y^2-6y+14\right) \\
&+z^2\!\left(y^4+3y^3+36y^2-20y+28\right)
+z^3\!\left(10y^3+49y^2-10y+22\right) \\
&+z^4\!\left(11z^4-2z^3-z^2-2z+8\right)
+F_4(z,y).
\end{aligned}
\]

We have that
\[
F_3(z, \alpha) = F_5(z, (\alpha - 1)z).
\]
Again, see this \href{https://sagecell.sagemath.org/?z=eJxtksFu2zAMhu8B8g5cdpHduLAlW7GDdbsMOa5rDjsVAbzVKQRkTue02Oxh7z5SpBKnqS8SyI-_ftJcX38Y5tDPof4IN_B1v-vb_U9X79aufVR3d3Mw0XQynaxyzKphk8fKxP0mhyso8TR42gVeNF7wRBbCdwXDpohVVglXWOGMuQTtSXcRdMvA20t-gcJWQNLzYFFdgmWsWFYHuhJav0FXsSqEO76OfvDIX6NZihZSgUrPZNkFlMWq8jl9njIxJjWPtqDRep8aEngFk_GBcuSJrSdg-bl8jA0bHVo1oVUrvE7ZQ3leYKQDYvOgnQmrz1n88dgMniSHp0nwwQRv1P0YXeW-q_dw27lH19Y7aP48dc3h4PYtNmuwWWo9rr1VuZrTVXPLNDqO-1Glwh8zhOUhY8QUZfyYJM5bSb_1JFV6pQQKgWiQWpjAl6KJK-YTvC0c599NGxhs22MvFdYYzSA9p4WjGpKyUkPGS_IgAIWM5LX4LIRFj8ggvCiYIDgVIAwrF1iTskW41EyQWCqAbDPGSc1ylF6gdowA3pKfj-WotBt-jSU1P2-Jkh1_pVdxVWtePNqLzPImfGl-ny0B3XENVsX14eX7Qf3tl6BqGncUD__89lDVt6Zz2x6aXy_1zj3308mD226brml_NFi8VqhStw_KiyWwMpGvfOpc-6xm_AS-YeDTbA7j0htIPem24_A7DC95jUXivv18TC9n0Th3qov-A-8LN-Q=&lang=sage&interacts=eJyLjgUAARUAuQ==}{link} for a computer verification of this expression.
Finally, we prove that \(F_5(z, y) \geq 0\)
for all \(z, y \in \Reals_+\).

The polynomial \(F_4(z, y)\) is trivially non-negative when \(z, y \geq 0\) since all coefficients are nonnegative.
We now analyze term by term of \(F_5\):
\begin{itemize}
    \item \(y^2 - y + 2\),

    it is a parabola with positive leading coefficient and negative determinant, hence always positive. 
    
    \item \(y^3 + 9y^2 - 6y + 14\),
    
    it is at least \(9y^2 - 6y + 14\), which is a parabola with positive leading coefficient and negative determinant, hence always positive.

    \item \(y^4+3y^3+36y^2-20y+28\),

    it is at least \(36y^2-20y+28\), which is a parabola with positive leading coefficient and negative determinant, hence always positive.

    \item \(10y^3+49y^2-10y+22\),

    it is at least \(49y^2-10y+22\), which is a parabola with positive leading coefficient and negative determinant, hence always positive.

    \item \(11z^4-2z^3-z^2-2z+8\),

    if \(z \leq 1\), then
    \(11z^4-2z^3-z^2-2z+8 \geq 0 -2 -1 -2 + 8 = 3 > 0\). If \(z \geq 1\), then
    \(11z^4-2z^3-z^2-2z+8 \geq 11z^4-2z^4-z^4-2z^4 + 8 =6z^4 + 8 > 0\).
\end{itemize}

We conclude that \(F_5(z, y), F_4(z, y), F_3(z, \alpha), F_2(z, \alpha), F_1(z, \alpha)\) are nonnegative for all values of \(z, y, \alpha\) in their domain. Consequently, our main inequality is proved.

\section{Integrality Gap}
\label{sec:appendixIG}
In this section, we provide an instance showing the integrality gap is at most \(3/4\) and overview a construction leading to better bounds.

Consider an instance with \(N\) players and \(1\) item, each player has two functions in their distribution support.
For player \(i\), with probability \(1/2\) they sample the function \(v_{i,1}(x) \coloneqq \min \{x, 2^{-i+1}\}\) and with probability \(1/2\) they sample \(v_{i,2}(x) = 0\).
The LP solution \(x_{i, 1, 1} = 2^{-i}\) and \(x_{i,2,0} = 0\) for every \(i \in [N]\) is feasible, satisfying the online constraint with equality, and achieves objective value
\[
\sum_{i = 1}^N \frac12 \min\left\{\frac{2^{-i}}{p_i},2^{-i+1} \right\} = \sum_{i = 1}^N 2^{-i} = 1 - \frac1{2^{N}}.
\]
On the other hand, for the online algorithm, note that for every player that samples the nonzero function it is optimal to give the item until the cap is reached.
Therefore, we have
\[
\E[\OPT] = \Pr[V_1 = v_{1, 1}] \times 1 +
\Pr[V_1 = 0] \times \E[\OPT | V_1 = 0].
\]
Note that if \(V_1 = 0\), \(\OPT\) will assign \(0\) to the first player, and every other player that samples a nonzero function will receive up to its limit. Thus,
\[
\E[\OPT] = \frac12 + \frac12 \sum_{i = 2}^N \frac12 \times\frac1{2^{i-1}}
=
\frac12 + \frac14\left(1 - \frac{1}{2^{N-1}}\right)
=
\frac34 - \frac1{2^{N+1}}.
\]
Thus, by taking \(N \to \infty\), we obtain an integrality gap of \(3/4\).

We remark that there are finite instances achieving integrality gap exactly \(3/4\). However, the previous example leads to the following generalization.
Let \(m \geq 1\) be an integer, consider an instance with
\(N = m^2\) players and \(1\) item.
Each player has two functions in their distribution support.
Player \(i\) samples the function \(v_{i, 1}(x) = \min \{x, (1-m^{-1})^{i-1}\}\) with probability \(1/m\) and \(v_{i,2} = 0\) with the remaining probability.
Consider the feasible solution \(x_{i, 1} = \tfrac{1}{m} \left(1 - \tfrac1{m}\right)^{i - 1}\) and \(x_{i,2} = 0\)
for every \(i \in [N]\).
Then, the objective value of \(x\) goes to \(1\) as \(m \to \infty\).
Like previously, as soon as a player arrives the optimal algorithm should greedily assign as much as possible until it reaches the cap.
One can show that when \(m \to \infty\), the problem converges to a Poisson Process where an arrival at time \(t\)
provides a reward of value \(\min\{R(t), e^{-t}\}\) where \(R(t)\) is the remainder of the item or \(1\) minus how much was assigned in the previous events. Such process has an expected value of \(1 - e^{-\gamma}/2 \approx 0.719\) where \(\gamma\) is the Euler-Mascheroni constant.

{\small
\bibliography{sample-bibliography}
\bibliographystyle{alpha}
}

\clearpage
\end{document}